\documentclass[11pt,letterpaper]{article}

\def\showauthornotes{0}
\def\showtableofcontents{1}
\def\showkeys{0}
\def\showdraftbox{0}
\def\showcolorlinks{1}
\def\usemicrotype{0}
\def\showfixme{0}

\def\writemode{0}

\usepackage{etex}

\usepackage[l2tabu, orthodox]{nag}

\usepackage{xspace,enumerate}

\usepackage[dvipsnames]{xcolor}

\usepackage[T1]{fontenc}
\usepackage[full]{textcomp}

\usepackage[american]{babel}

\usepackage{mathtools}

\usepackage{amsthm}

\newtheorem{theorem}{Theorem}[section]
\newtheorem*{theorem*}{Theorem}

\newtheorem*{proposition*}{Proposition}
\newtheorem{lemma}[theorem]{Lemma}
\newtheorem*{lemma*}{Lemma}
\newtheorem{corollary}[theorem]{Corollary}
\newtheorem*{conjecture*}{Conjecture}
\newtheorem{fact}[theorem]{Fact}
\newtheorem*{fact*}{Fact}

\newtheorem*{hypothesis*}{Hypothesis}

\theoremstyle{definition}
\newtheorem{definition}[theorem]{Definition}

\newtheorem{question}[theorem]{Question}

\newtheorem{problem}[theorem]{Problem}

\theoremstyle{remark}
\newtheorem{claim}[theorem]{Claim}
\newtheorem*{claim*}{Claim}
\newtheorem{remark}[theorem]{Remark}
\newtheorem*{remark*}{Remark}

\newtheorem*{observation*}{Observation}

\ifnum\writemode=1
\usepackage[
letterpaper,
top=1.2in,
bottom=1.2in,
left=1in,
right=1in]{geometry}

\fi

\ifnum\writemode=0
\usepackage[
letterpaper,
top=0.7in,
bottom=0.9in,
left=1in,
right=1in]{geometry}
\fi

\usepackage{amssymb,amsmath}

\let\bm\mathbf
\ifnum\showkeys=1
\usepackage[color]{showkeys}
\fi

\ifnum\showcolorlinks=1
\usepackage[
pagebackref,
colorlinks=true,
urlcolor=blue,
linkcolor=blue,
citecolor=OliveGreen,
]{hyperref}
\fi

\ifnum\showcolorlinks=0
\usepackage[
pagebackref,
colorlinks=false,
pdfborder={0 0 0}
]{hyperref}
\fi

\newcommand{\Sref}[1]{\hyperref[#1]{\S\ref*{#1}}}

\usepackage{nicefrac}
\ifnum\usemicrotype=1
\usepackage{microtype}
\fi

\ifnum\showauthornotes=1
\newcommand{\Authornote}[2]{{\sffamily\small\color{red}{[#1: #2]}}}
\newcommand{\Authornotecolored}[3]{{\sffamily\small\color{#1}{[#2: #3]}}}
\newcommand{\Authorcomment}[2]{{\sffamily\small\color{gray}{[#1: #2]}}}
\newcommand{\Authorstartcomment}[1]{\sffamily\small\color{gray}[#1: }

\newcommand{\Authorfnote}[2]{\footnote{\color{red}{#1: #2}}}
\newcommand{\Authorfixme}[1]{\Authornote{#1}{\textbf{??}}}
\newcommand{\Authormarginmark}[1]{\marginpar{\textcolor{red}{\fbox{\Large #1:!}}}}
\else
\newcommand{\Authornote}[2]{}
\newcommand{\Authornotecolored}[3]{}
\newcommand{\Authorcomment}[2]{}
\newcommand{\Authorstartcomment}[1]{}

\newcommand{\Authorfnote}[2]{}
\newcommand{\Authorfixme}[1]{}
\newcommand{\Authormarginmark}[1]{}
\fi

\newcommand{\Dnote}{\Authornote{D}}

\newcommand{\Pnote}{\Authornote{P}}

\newcommand{\Bnote}{\Authornote{B}}

\definecolor{forestgreen(traditional)}{rgb}{0.0, 0.27, 0.13}

\ifnum\showfixme=0

\fi

\usepackage{boxedminipage}

\newcommand{\Paren}[1]{\left(#1\right)}

\newcommand{\abs}[1]{\lvert#1\rvert}

\newcommand{\Set}[1]{\left\{#1\right\}}

\newcommand{\norm}[1]{\lVert#1\rVert}
\newcommand{\Norm}[1]{\left\lVert#1\right\rVert}

\newcommand{\iprod}[1]{\langle#1\rangle}

\newcommand{\Esymb}{\mathbb{E}}
\newcommand{\Psymb}{\mathbb{P}}

\DeclareMathOperator*{\E}{\Esymb}

\DeclareMathOperator*{\ProbOp}{\Psymb}

\renewcommand{\Pr}{\ProbOp}

\newcommand{\textparen}[1]{\text{(#1)}}

\ifx\because\undefined
\newcommand{\because}[1]{\textparen{because #1}}
\else
\renewcommand{\because}[1]{\textparen{because #1}}
\fi

\newcommand{\mper}{\,.}
\newcommand{\mcom}{\,,}

\newcommand\bdot\bullet

\ifx\mathds\undefined % use double stroke fonts if available

\else
}
\fi

\DeclareMathOperator{\Tr}{Tr}

\DeclareMathOperator{\poly}{poly}

\DeclareMathOperator{\polylog}{polylog}

\newcommand{\Hoelder}{H\"{o}lder\xspace}
\newcommand{\Holder}{\Hoelder}
\newcommand{\Brandao}{Brand\~ao\xspace}

\newcommand{\N}{\mathbb N}
\newcommand{\R}{\mathbb R}
\newcommand{\C}{\mathbb C}

\newcommand{\problemmacro}[1]{\texorpdfstring{\textup{\textsc{#1}}}{#1}\xspace}

\newcommand{\uniquegames}{\problemmacro{unique games}}

\newcommand{\smallsetexpansion}{\problemmacro{small-set expansion}}

\newcommand{\cA}{\mathcal A}

\newcommand{\cE}{\mathcal E}

\newcommand{\cM}{\mathcal M}

\newcommand{\cW}{\mathcal W}

\newcommand{\cY}{\mathcal Y}

\newcommand{\bbS}{\mathbb S}

\renewcommand{\leq}{\leqslant}
\renewcommand{\le}{\leqslant}
\renewcommand{\geq}{\geqslant}
\renewcommand{\ge}{\geqslant}

\ifnum\showdraftbox=1
\newcommand{\draftbox}{\begin{center}
  \fbox{%
    \begin{minipage}{2in}%
      \begin{center}%
          \Large\textsc{Working Draft}\\%
        Please do not distribute%
      \end{center}%
    \end{minipage}%
  }%
\end{center}
\vspace{0.2cm}}
\else
\newcommand{\draftbox}{}
\fi

\let\epsilon=\varepsilon

\numberwithin{equation}{section}

\newcommand\MYcurrentlabel{xxx}

\newcommand{\MYstore}[2]{%
  \global\expandafter \def \csname MYMEMORY #1 \endcsname{#2}%
}

\newcommand{\MYload}[1]{%
  \csname MYMEMORY #1 \endcsname%
}

\newcommand{\MYnewlabel}[1]{%
  \renewcommand\MYcurrentlabel{#1}%
  \MYoldlabel{#1}%
}

\newcommand{\MYdummylabel}[1]{}

\newcommand{\torestate}[1]{%
  \let\MYoldlabel\label%
  \let\label\MYnewlabel%
  #1%
  \MYstore{\MYcurrentlabel}{#1}%
  \let\label\MYoldlabel%
}

\newcommand{\restatetheorem}[1]{%
  \let\MYoldlabel\label
  \let\label\MYdummylabel
  \begin{theorem*}[Restatement of \prettyref{#1}]
    \MYload{#1}
  \end{theorem*}
  \let\label\MYoldlabel
}

\newcommand{\restatelemma}[1]{%
  \let\MYoldlabel\label
  \let\label\MYdummylabel
  \begin{lemma*}[Restatement of \prettyref{#1}]
    \MYload{#1}
  \end{lemma*}
  \let\label\MYoldlabel
}

\newcommand{\restateprop}[1]{%
  \let\MYoldlabel\label
  \let\label\MYdummylabel
  \begin{proposition*}[Restatement of \prettyref{#1}]
    \MYload{#1}
  \end{proposition*}
  \let\label\MYoldlabel
}

\newcommand{\restatefact}[1]{%
  \let\MYoldlabel\label
  \let\label\MYdummylabel
  \begin{fact*}[Restatement of \prettyref{#1}]
    \MYload{#1}
  \end{fact*}
  \let\label\MYoldlabel
}

\newcommand{\restate}[1]{%
  \let\MYoldlabel\label
  \let\label\MYdummylabel
  \MYload{#1}
  \let\label\MYoldlabel
}

\newcommand{\addreferencesection}{
  \phantomsection
  \addcontentsline{toc}{section}{References}
}

\newcommand{\e}{\epsilon}

\let\origparagraph\paragraph
\renewcommand{\paragraph}[1]{\origparagraph{#1.}}

\allowdisplaybreaks

\newcommand{\cclassmacro}[1]{\texorpdfstring{\textbf{#1}}{#1}\xspace}

\newcommand{\np}{\cclassmacro{NP}}

\usepackage{paralist}

\usepackage{comment}

\usepackage{braket}

\usepackage{relsize}
\usepackage[font=footnotesize]{caption}
\usepackage{appendix}

\newcommand{\sphere}{\bbS^{n-1}}

\newcommand{\Sn}{\mathbb{S}^{n-1}}
\DeclareMathOperator{\Id}{\mathrm{Id}}
\DeclareMathOperator{\F}{\mathbb{F}}
\DeclareMathOperator{\codim}{codim}

\DeclareUrlCommand\email{}

\DeclareMathOperator*{\pE}{\tilde{\mathbb E}}

\newcommand{\tr}{\textup{tr}}

\newcommand*{\transpose}[1]{{#1}{}^{\mkern-4mu\intercal}}

\newcommand*{\dyad}[1]{#1#1{}^{\mkern-4mu\intercal}}

\newcommand{\1}{\bm{1}}

\title{Quantum entanglement, sum of squares, and the log rank conjecture}

\author{%
Boaz Barak\thanks{Harvard Paulson School, \protect\email{b@boazbarak.org}. Supported by NSF awards CCF 1565264 and CNS 1618026.}
\and
Pravesh K. Kothari \thanks{Princeton University and IAS \protect \email{kothari@cs.princeton.edu}. Part of the work was done while visiting Harvard University.}
\and
David Steurer\thanks{Cornell University and Institute for Advanced Study, \protect\email{dsteurer@cs.cornell.edu}.
Supported by a Microsoft Research Fellowship, a Alfred P. Sloan Fellowship, an NSF awards (CCF-1408673,CCF-1412958,CCF-1350196), and the Simons Collaboration for Algorithms and Geometry.}}

\begin{document}

\maketitle
 \draftbox
\thispagestyle{empty}

\begin{abstract}
For every $\epsilon>0$, we give an $\exp(\tilde{O}(\sqrt{n}/\epsilon^2))$-time algorithm for the $1$ vs $1-\epsilon$ \emph{Best Separable State (BSS)} problem of distinguishing,
given an $n^2\times n^2$ matrix $\cM$ corresponding to a quantum measurement, between the case that there is a separable (i.e., non-entangled) state $\rho$ that $\cM$ accepts with probability $1$, and the case  that every separable state is accepted with probability at most $1-\epsilon$.
Equivalently, our algorithm takes the description of a subspace $\cW \subseteq \mathbb F^{n^2}$ (where $\mathbb F$ can be either the real or complex field) and distinguishes between the case that $\cW$ contains a rank one matrix, and the case that every rank one matrix is at least  $\epsilon$ far (in $\ell_2$ distance) from $\cW$. 

To the best of our knowledge, this  is the first improvement over the brute-force $\exp(n)$-time algorithm for this problem. 
Our algorithm is based on the \emph{sum-of-squares} hierarchy and its analysis is inspired by Lovett's proof (STOC '14, JACM '16) that the communication complexity of every rank-$n$ Boolean matrix is bounded by $\tilde{O}(\sqrt{n})$. 

\iffalse
%%%% Markdown version below

For every constant $\epsilon>0$, we give an $\exp(\tilde{O}(\sqrt{n}))$-time algorithm for the $1$ vs $1-\epsilon$ *Best Separable State (BSS)* problem of distinguishing,
given an $n^2\times n^2$ matrix $M$ corresponding to a quantum measurement, between the case that there is a separable (i.e., non-entangled) state $\rho$ that $M$ accepts with probability $1$, and the case  that every separable state is accepted with probability at most $1-\epsilon$.
Equivalently, our algorithm takes the description of a subspace $W \subseteq \mathbb{F}^{n^2}$ (where $\mathbb{F}$ can be either the real or complex field) and distinguishes between the case that $W$ contains a rank one matrix, and the case that every rank one matrix is at least  $\epsilon$ far (in $\ell_2$ distance) from $W$. 

To the best of our knowledge, this  is the first improvement over the brute-force $\exp(n)$-time algorithm for this problem. 
Our algorithm is based on the *sum-of-squares* hierarchy and its analysis is inspired by Lovett's proof (STOC '14, JACM '16) that the communication complexity of every rank-$n$ Boolean matrix is bounded by $\tilde{O}(\sqrt{n})$. 
\fi
\end{abstract}

\clearpage

% tableofcontents added for better navigability of the document
\ifnum\showtableofcontents=1
{
\tableofcontents
\thispagestyle{empty}
 }
\fi

\clearpage

\setcounter{page}{1}
%!TEX root = ../subexpalg.tex
\newcommand{\bss}{\problemmacro{best separable state}}
\newcommand{\BSS}{\problemmacro{BSS}}
\newcommand{\qma}{\cclassmacro{QMA}}
\newcommand{\ma}{\cclassmacro{MA}}
\newcommand{\EXP}{\cclassmacro{EXP}}
\newcommand{\nexp}{\cclassmacro{NEXP}}
\newcommand{\on}{\{-1,1\}}
\newcommand{\zo}{\{0,1\}}

\section{Introduction}
\Pnote{TODO:
1. Complete proof of Structure Theorem's Corollary (some parts are commented out in the proof of structure theorem)
2. Add blurb about complex SOS and appendix about proving the complex version of the structure theorem. 
3. Clean up scalar reweighting and spell out how "univariate pseudo-distributions can always be extended" fact is used. 
4. Add references.}
\emph{Entanglement} is one of the more mysterious and subtle phenomena in quantum mechanics.
The formal definition is below (Definition~\ref{def:separable}), but roughly speaking, a joint quantum state $\rho$ of two sub-systems $A$ and $B$ is \emph{entangled} if a quantum measurement of one system can affect the other system in a way that cannot be captured using classical correlations. 
A non-entangled state is called \emph{separable}.
Entanglement has often been talked of as "spooky interaction at a distance" and is  responsible for many of the more counter-intuitive features of quantum mechanics.  
It is also a crucial aspect of quantum algorithms that obtain speedups over the best known classical algorithms, and it may be necessary for such speedups~\cite{PhysRevLett.91.147902}.

One of the indicators of the underlying complexity of entanglement is that even given the full description of a quantum state $\rho$ as a density matrix, there is no known efficient algorithm for determining whether $\rho$ is entangled or not.
Indeed, the best known algorithms take time which is \emph{exponential} in the dimension of the state (which itself is exponential in the number of underlying qubits). 
This is in contrast to the classical case, where there is an efficient algorithm for the analogous problem of finding whether a given probability distribution $\mu$ over a universe $A\times B$ is a \emph{product distribution} which can be done by simply computing the rank of the PDF of $\mu$ when viewed as a matrix. 

Given the inherently probabilistic and noisy setting of quantum computing, an arguably better motivated question is the robust version of distinguishing between the case that a state $\rho$ is separable, and the case that it is $\epsilon$-far from being separable, in the sense that there exists some \emph{measurement}  $\cM$ that accepts $\rho$ with probability $p$ but accepts every separable state with probability at most $p-\epsilon$. 
This problem is known as the \emph{Quantum Separability Problem} with parameter $\epsilon$. 
Gharibian~\cite{gharibian2010strong}, improving on Gurvits~\cite{gurvits2003classical}, showed that this problem is NP hard when $\epsilon$ is inversely polynomial in the dimension of the state.
Harrow and Montanaro~\cite{DBLP:journals/jacm/HarrowM13} showed that, assuming the Exponential Time Hypothesis, there is no $n^{o(\log n)}$ time algorithm for this problem for $\epsilon$ which is a small constant.

\Bnote{(Update: looking at their section 4.2, it seems that we're quoting  Harrow and Montanaro fine - let me know if you think it's different.}

A closely related problem, which is the one we focus on in this paper, is the \emph{Best Separable State (BSS)} problem.\footnote{Using the connection between optimization and separation oracles in convex programming, one can convert  a sufficiently good algorithm for the search variant of one of these problems to the other. See \cite[Sec.~4.2]{DBLP:journals/jacm/HarrowM13} for a thorough discussion of the relations between these and many other problems.} In the BSS problem, the input is a measurement $\cM$ on a two part system and two numbers $1 \geq c > s \geq 0$ and the goal is to distinguish between the YES case that there is a separable state that $\cM$ accepts with probability at least $c$ and the NO case that $\cM$ accepts every separable state with probability at most $s$.
In particular, certifying that a particular measurement $\cM$ satisfies the NO case is extremely useful since it implies that $\cM$ can serve as an \emph{entanglement witness}~\cite{horodecki1996separability,lewenstein2000optimization}, in the sense that achieving acceptance probability with $\cM$ larger than $s$ certifies the presence of entanglement in a state. 
Such entanglement witnesses are used to certify entanglement in experiments and systems such as candidate computing devices~\cite{vedral2008quantifying}, and so having an efficient way to certify that they are sound (do not accept separable states) can be extremely useful.

Analogous to the quantum separability problem,  the BSS  problem is NP hard when $c-s = 1/poly(n)$~\cite{blier2009all,gurvits2003classical} and Harrow and Montanaro~\cite[Corollary 13(i)]{DBLP:journals/jacm/HarrowM13} show that (assuming the ETH) there is no $n^{o(\log n)}$ time algorithm for $\BSS_{1,1/2}$. 
An outstanding open question is whether the \cite{DBLP:journals/jacm/HarrowM13} result is \emph{tight}: whether there is a quasi-polynomial time algorithm for $\BSS_{c,s}$ for some constants $1 \geq c > s \geq 0$.
This question also has a quantum complexity interpretation. 
A measurement on a two part system can be thought of as a \textit{verifier} (with hardwired input) that interacts with two provers. 
Requiring the state to be \textit{separable} corresponds to stipulating that the two provers are not entangled. 
Thus it is not hard to see that an algorithm for $\BSS_{c,s}$ corresponds to an algorithm for deciding all languages in the complexity class $QMA(2)$ of \emph{two prover quantum Merlin Arthur}  systems with corresponding completeness and soundness parameters $c$ and $s$ respectively. 
In particular, a quasi-polynomial time algorithm for $\BSS_{0.99,0.5}$ would imply that $QMA(2) \subseteq EXP$, resolving  a longstanding problem in quantum complexity.\footnote{For more on information on this problem and its importance, see the presentations in the recent  workshop \url{http://qma2016.quics.umd.edu/} that was dedicated to it.}

In 2004, Doherty, Parrilo and Spedalieri~\cite{doherty2004complete} proposed an algorithm for the BSS problem  based on the \emph{Sum of Squares} semidefinite programming hierarchy~\cite{parrilo2000structured,DBLP:conf/ipco/Lasserre01}. 
It is not known whether this algorithm can solve the $\BSS_{c,s}$ problem (for constants $c>s$) in quasi-polynomial time. 
However \Brandao, Christandl and Yard~\cite{Brandao:2011:QAQ:1993636.1993683}  showed that it runs in quasi-polynomial time when the measurement $\cM$ is restricted to a special class of measurements known as \textit{one-way local operations and classical communications} (1-LOCC). 
\Brandao and Harrow~\cite{DBLP:journals/corr/BrandaoH15} showed that similar performance for these types of measurements can be achieved by an algorithm based on searching on an appropriately defined  $\epsilon$-net.

\subsection{Non quantum motivations}

The BSS problem is actually quite natural and well motivated from classical considerations.
As we'll see in Section~\ref{sec:techniques} below, it  turns out that at its core lies the following problem:

\begin{definition}[Rank one vector in subspace problem] \label{def:rank-one} 
Let $\F\in \{\R,\C\}$ and $\epsilon>0$. 
The \emph{$\epsilon$ rank one vector problem over $\F$}  is the task of distinguishing,  given a linear subspace $\cW \subseteq \mathbb{F}^{n^2}$, between the case that there is a nonzero rank one matrix $L\in \cW$ and the case that $\norm{L-M}_F \geq \epsilon\norm{L}_F$  for every rank one $L$ and $M\in\cW$.\footnote{For a $k\times m$ matrix $A$, we denote by $\norm{A}_F$ its \emph{Frobenius} norm, defined as $\sqrt{\sum_{i,j}|A_{i,j}|^2} = \Tr(AA^*)^{1/2}$, which is the same as taking the $\ell_2$ norm of the matrix when considered as an $km$-dimensional vector.}
\end{definition}

\Dnote{previous definition allowed a general field. but then it is unclear how to define something like the Frobenius norm.}

This is arguably a natural problem in its own right. 
While solving this problem exactly (i.e., determining if there is a rank one solution to a set of linear equations) is the same as the NP hard task of solving \textit{quadratic equations}, it turns out that we can obtain non-trivial algorithmic results by considering the above notion of approximation.
Indeed, our main result implies an $\exp(\tilde{O}(\sqrt{n}))$ time algorithm for this problem for any constant $\epsilon>0$ in both the real and complex cases.

\subsection{Our results} \label{sec:results}

In this work we give a $2^{\tilde{O}(\sqrt{n})}$ time algorithm for the $\BSS_{1,s}$ problem for every constant $s<1$. 
We now make the necessary definitions and state our main result.\footnote{For the sake of accessibility, as well as to emphasize the connections with non-quantum questions, we use standard linear algebra notation in this paper as opposed to Dirac's ket notation that is more common in quantum mechanics. A vector $u$ is a column vector unless stated otherwise, and $u^*$ denotes the complex conjugate transpose of the vector $u$. If $u$ is real, then we denote its transpose by $u^\top$. See the lecture notes \cite{brandao2016mathematics} for a more complete coverage of separability and entanglement.}

% https://arxiv.org/abs/1604.01790

\begin{definition} \label{def:separable} A \emph{quantum state} on a system of $m$ elementary states (e.g., a $\log m$-qubit register) is an $m\times m$ complex Hermitian matrix $\rho$ (known as a \emph{density matrix}) such that $\Tr \rho = 1$.
A quantum state $\rho$ is \emph{pure} if it is of the form $\rho = ww^*$ for some unit vector $w\in \C^m$. Otherwise we say that $\rho$ is \emph{mixed}. 
Note that every mixed state $\rho$ is a convex combination of pure states.

If $m=n^2$, and we identify $[m]$ with $[n]\times [n]$ then an $m$-dimensional pure quantum state $\rho = ww^* \in \C^{m^2}$ is \emph{separable} if the vector $w\in\C^m$ is equal to  $uv^*$ for some $u,v \in \C^n$. A general state $\rho$ is \emph{separable} if it is a convex combination of separable pure states. That is,
$\rho = \E (uv^*)(uv^*)^*$
where the expectation is taken over a  distribution supported over pairs of unit vectors $u,v \in \C^n$.
A state that is not separable is called \emph{entangled}.

A quantum \emph{measurement operator} is an $m\times m$ complex Hermitian matrix $\cM$ such that $0 \preceq \cM \preceq I$. The probability  that a measurement $\cM$ accepts a state $\rho$ is $\Tr(\rho \cM)$.  
\end{definition}
\Bnote{Seems that in quantum texts they write $\Tr(\rho \cM)$ rather than the other way around, I think we don't need the star since $\cM=\cM^*$.}

\begin{theorem}[Main result] \label{thm:main-thm}
For every $s<1$, there is a $2^{\tilde{O}(\sqrt{n})}$ time algorithm, based on $\tilde{O}(\sqrt{n})$ rounds of the sos hierarchy, that on input an $n^2\times n^2$ measurement operator $\cM$, distinguishes between the following two cases:
\begin{itemize}
\item YES: There exists a separable state $\rho \in \C^{n^2\times n^2}$ such that $\Tr(\rho \cM)=1$.
\item NO: For every separable $\rho  \in \C^{n^2\times n^2}$, $\Tr(\rho\cM) \leq s$
\end{itemize}
\end{theorem}

To our knowledge, this algorithm is the first for this problem that beats the brute force bound of $2^{O(n)}$ time for general measurements.

Like the algorithms of \cite{doherty2004complete,Brandao:2011:QAQ:1993636.1993683}, our algorithm is based on the \textit{sum of squares} SDP hierarchy, but we introduce new techniques for analyzing it that we believe are of independent interest. 
As we discuss in Section~\ref{sec:conclusions}, it  is a fascinating open question to explore whether our techniques can be quantitatively strengthened to yield faster algorithms and/or extended for other problems such as the $2$ to $4$ norm and  small set expansion, that have been shown to be related to the BSS problem by~\cite{DBLP:conf/stoc/BarakBHKSZ12} (albeit in a different regime of parameters than the one we deal with in this work). 
As we remark below, this question seems related to other longstanding open questions in computer science and in particular to the \textit{log rank conjecture} in communication complexity~\cite{DBLP:conf/focs/LovaszS88}.

\begin{remark}[Imperfect completeness] \label{rem:perfect-completeness}
We state our results for the case of perfect completeness for simplicity, but all of the proofs extend to the case of ``near perfect completeness'' where in the YES case we replace the condition $\Tr(\rho\cM)=1$ with the condition $\Tr(\rho\cM) = 1 - \tfrac{1}{n}$ (see  Remark \ref{rem:imperfect-completeness}). 
It is an interesting open problem to find out whether our  results can extend to the setting where in the YES case $\Tr(\rho\cM)=1-\epsilon$ for some absolute constant $\epsilon$.
We conjecture that this is indeed the case.
\end{remark}

\begin{remark}[Real vs complex numbers]\label{rem:reals}
While the natural setting for quantum information theory is the \textit{complex numbers}, much of the power and interest already arises in the case of the real numbers, which is more natural for the sos algorithm (though it does have complex-valued generalization). For our purposes, there's no difference between the real and the complex cases - we give a reduction from the complex case to the real case in Section \ref{app:red-complex-to-real} of the Appendix. Thus, from now on, we will focus solely on the case that all operators, subspaces, matrices are \emph{real}. 
%We note that there is a natural mapping of an $n\tim{}es n$ complex matrix $A+iB$ (with  real $n\times n$ matrices $A,B$)into the $2n\times 2n$ real matrix $\left( \begin{smallmatrix} A&B\\ -B&A \end{smallmatrix} \right)$. Note that a complex rank one decomposition $A+iB = (x+iy)(z+iw)^*$ for $x,y,z,w\in \R^n$ will translate into a rank two decomposition $\left( \begin{smallmatrix} x&y\\ -y&x \end{smallmatrix} \right)\left( \begin{smallmatrix} z&w\\ -w&z \end{smallmatrix} \right)^*$. We can use a  higher dimensional version of our result (see Theorem~\ref{thm:multi-dimensional-rank1}) to find such decompositions, though we defer the complete derivation to the full version of this paper. 
\end{remark}

\section{Our techniques} \label{sec:techniques}

Our algorithm follows a recent paradigm of constructing rounding algorithms for the sum of squares sdp by considering its solutions as "pseudo-distributions"~\cite{sos-lecture-notes}.
These can be thought of as capturing the uncertainty that a computationally bounded solver has about the optimal solution of the given problem, analogous to the way that probability distributions model uncertainty in the classical information-theoretic Bayesian setting.

Somewhat surprisingly, our main tool in analyzing the algorithm are techniques that arose in  proof of the currently best known upper bound for the \textit{log rank conjecture}~\cite{DBLP:conf/focs/LovaszS88}. 
This conjecture has several equivalent formulations, one of which is that  every $N\times N$ matrix $A$ with Boolean (i.e., $0/1$) entries and rank at most $n$, contains a submatrix of size at least $2^{-\poly\log(n)}N \times 2^{-\poly\log(n)}N$ that is of rank one.\footnote{The original formulation of the log rank conjecture is that every such matrix has communication complexity at most $\poly\log(n)$, and Nisan and Wigderson~\cite{DBLP:conf/focs/NisanW94} showed that this is equivalent to the condition that such matrices contains a monochromatic submatrix of the above size. Every monochromatic submatrix is rank one, and every rank one submatrix of size $s\times s$ of a Boolean valued matrix contains a monochromatic submatrix of size at least $\tfrac{s}{2} \times \tfrac{s}{2}$.} 
The best known bound on the log rank conjecture is by Lovett~\cite{DBLP:conf/stoc/Lovett14} who proved that every such matrix contains a submatrix of size at least $2^{-\tilde{O}(\sqrt{n})}N \times 2^{-\tilde{O}(\sqrt{n})}N$.

Our algorithm works by combining the following observations:

\begin{enumerate}

\item Lovett's proof can be generalized to show that \textit{every} $N\times N$ rank $n$ real (or complex) matrix $A$ (not necessarily with Boolean entries)  contains a $2^{-\tilde{O}(\sqrt{n})}N\times 2^{-\tilde{O}(\sqrt{n})}N$ submatrix that is \textit{close} to rank one in Frobenius norm.

\item If $\mu$ is an \textit{actual} distribution over solutions to the sos program for the BSS problem on dimension $n$, then we can transform $\mu$ into an $N\times N$ rank $n$ matrix $A=A(\mu)$ such that extracting an approximate solution from $A$ in time $2^{\tilde{O}(k)}$ can be done if $A$ contains an approximately rank one submatrix of size at least $2^{-k}N\times 2^{-k}N$. 

\item Moreover all the arguments used to establish steps 1 and 2 above can be encapsulated in the sum of squares framework, and hence yield an algorithm that extracts an approximately optimal solution to the BSS problem  from a degree $\tilde{O}(\sqrt{n})$ pseudo-distribution $\mu$ that "pretends" to be supported over exact solutions. 

\end{enumerate}

Thus, even though in the sos setting there is no actual distribution $\mu$, and hence no actual matrix $A$, we can still use structural results on this "fake" (or "pseudo") matrix $A$ to obtain an \textit{actual} rounding algorithm. 
We view this as a demonstration of the power of the "pseudo distribution" paradigm to help in the discovery of new algorithms, that might not seem as natural without placing them in this framework. 

\subsection{Rounding from rank one reweightings}
\label{sec:rounding-from-rank-one}

We now give a more detailed (yet still quite informal) overview of the proof. 
As mentioned above, we focus on the case that the $n^2 \times n^2$ measurement matrix $\cM$ is \textit{real} (as opposed to \textit{complex}) valued. 

Let $\cW \subseteq \R^{n^2}$ be the subspace of vectors $X$ such that $X^\top \cM X = \norm{X}^2$ (this is a subspace since $\cM \preceq I$ and hence $\cW$ is the  eigenspace of $\cM$ corresponding to the eigenvalue $1$). 
We pretend that the sos algorithm yields a distribution $\mu$ over rank one matrices of the form $X=uv^\top$ such that $X \in \cW$. 
When designing a rounding algorithm, we only have access to \textit{marginals} of $\mu$, of the form $\E_\mu f(X)$ for some "simple" function $f$ (e.g., a low degree  polynomial).  
We need to show that we can use such "simple marginals" of $\mu$ to extract a single rank one matrix $u_0v_0^\top$ that has large projection into $\cW$. 

We start with the following simple observation:

\begin{lemma} \label{lem:round-from-rankone} If $\mu$ is a distribution over matrices $X$ in a subspace $\cW \subseteq \R^{n^2}$ such that the expectation $\E_\mu X$ is approximately rank one, in the sense that $\norm{L-\E_\mu X}_F \leq \epsilon\norm{L}_F$ for some rank one matrix $L$, then  $\Tr (\cM\rho) \geq 1-2\epsilon^2$ where $\rho$ is the pure separable state $\rho = LL^\top/\norm{L}_F^2$.
\end{lemma}
\begin{proof} Since $\mu$ is supported over matrices in $\cW$, $\E_\mu X$ is in $\cW$. But this means that the $\ell_2$  (i.e., Frobenius) norm distance of $L$ to the subspace $\cW$ is at most $\epsilon\norm{L}_F$. Since $\Tr(XX^\top\cM) = \Tr(X^\top \cM X) = \norm{X}_F^2$ for every $X\in\cW$, the value $\Tr(LL^\top M)$ will be at least as large as the norm squared of the projection of $L$ to $\cW$. 
\end{proof}

In particular this means that if we were lucky and the condition of Lemma~\ref{lem:round-from-rankone}'s statement occurs, then it would be trivial for us to extract from the expectation $\E_\mu X$ (which is a very simple marginal) a rank one matrix that is close to $\cW$, and hence achieves probability $1-\epsilon$ in the measurement $\cM$.
Note that even if every matrix in the support of $\mu$ has unit norm, the matrix $L$ could be of significantly smaller norm. 
We just need that there is some dimension-one subspace on which the cancellations among these matrices are significantly smaller than the cancellations in the rest of the dimensions. 

Of course there is no reason we should be so lucky, but one power that the marginals give us is the ability to \textit{reweight} the original distribution $\mu$. In particular, for every "simple" non-negative function $\zeta:\R^{n^2}\rightarrow\R_+$, we can compute the marginal $\E_{\mu_\zeta} X$ where $\mu_\zeta$ is the distribution over matrices where $\Pr_{\mu_\zeta}[X]$ (or $\mu_\zeta(X)$ for short) is proportional to $\zeta(X)\mu(X)$.  
As such when working with the solutions to the degree $k$ sos algorithm, we are only able to reweight using functions $\zeta$ that are polynomials of degree at most $k$, but for the purposes of this overview, let us pretend that we can reweight using any function that is not too "spiky" and make the following definition:

\begin{definition} Let $\mu$ be a probability distribution. We say that a probability distribution $\mu'$
is a \textit{$k$-deficient reweighting} of $\mu$ if $\Delta_{KL}(\mu'\|\mu)\leq k$ where $\Delta_{KL}(\mu'\|\mu)$ denotes the Kullback-Leibler divergence of $\mu'$ and $\mu$, defined as $\E_{X\sim\mu'} \log(\mu'(X)/\mu(X))$. 
\end{definition}

At least at a "moral level", the following theorem shows that a $k$-deficient reweighting (for $k \ll n$) can be helpful to prove our main result:

\begin{theorem}[Rank one reweighting]\label{thm:rank-one-reweighting} 
Let $\mu$ be any distribution over rank one $n\times n$ matrices and $\epsilon>0$.
Then there exists an $\sqrt{n}\poly(1/\epsilon)$-deficient reweighting $\mu'$ of $\mu$ and a rank one matrix $L$ such that 
\[
\norm{L-\pE_{\mu'} X}_F\leq \epsilon\norm{L}_F
\]
\end{theorem}

One of the results of this paper is a proof of Theorem~\ref{thm:rank-one-reweighting} (see Section~\ref{sec:proof-overview}). 
It turns out that this can be done using ideas from the works on the log rank conjecture.

\subsection{From monochromatic rectangles to rank one reweightings}
\label{sec:intro:mono-to-reweight}

What does Theorem~\ref{thm:rank-one-reweighting} has to do with the log rank conjecture? 
To see the connection let us imagine that the distribution $\mu$ is \textit{flat} in the sense that it is a uniform distribution over rank one matrices $\{ u_1v_1^\top ,\ldots, u_Nv_N^\top \}$ (this turns out to be essentially without loss of generality) and consider the $n\times N$ matrices $U$ and $V$ whose columns are $u_1,\ldots,u_N$ and $v_1,\ldots,v_N$ respectively.
The $n\times n$ matrix $\pE_\mu u_iv_i^\top$ is proportional  to $UV^\top$. 
This matrix has the same spectrum (i.e., singular values) as the $N\times N$ matrix $U^\top V$. 
Hence, $UV^\top$ is  close to a rank one matrix if and only if $U^\top V$ is, since in both cases this happens when the square of the top singular value dominates the sum of the squares of the rest of the singular values.
Now a flat distribution $\mu'$ with $\Delta_{KL}(\mu'\|\mu)\leq k$ corresponds to the uniform  distribution over $\{ u_iv_i^\top \}_{i\in I}$
where $I \subseteq [N]$ satisfies $|I| \geq 2^{-k}N$. 
We can see that $\E_{\mu'} u_iv_i^\top$ will be approximately rank one if and only if the submatrix of $U^\top V$ corresponding to $I$ is approximately rank one.
Using these ideas it can be shown that Theorem~\ref{thm:rank-one-reweighting} is equivalent to the following theorem:\footnote{To show this formally we use the fact that by Markov, every distribution $\mu'$ with $\Delta_{KL}(\mu'\|U_{[N]})=\log N - H(\mu') = k$ is $\epsilon$-close to a distribution with  min entropy $\log N-O(k/\epsilon)$ and every distribution of the latter type is a convex combination of flat  distributions of support at least $N2^{-O(k/\epsilon)}$. \Bnote{check me on this}}

\begin{theorem}[Rank one reweighting---dual formulation]\label{thm:rank-one-reweighting-dual} 
Let $A$ be any $N\times N$ matrix of rank at most $n$. Then there exists a subset $I\subseteq [N]$ with with $|I| \geq \exp(-\sqrt{n}\poly(1/\epsilon))N$
and a rank one matrix $L$ such that 
\[
\norm{L-A_{I,I}}_F\leq \epsilon\norm{L}_F
\]
where $A_{I,I}$ is the submatrix corresponding to restricting the rows and columns of $A$ to the set $I$. 
\end{theorem}

One can think of Theorem~\ref{thm:rank-one-reweighting-dual} as an approximate and robust version of Lovett's result~\cite{DBLP:conf/stoc/Lovett14} mentioned above.
Lovett showed that every $N\times N$ matrix of rank $n$ with \textit{Boolean} entries has a $2^{-\tilde{O}(\sqrt{{n}})}N\times 2^{-\tilde{O}(\sqrt{{n}})}N$  submatrix that is of exactly rank 1.
We show that the condition of Booleanity is  not needed if one is willing to relax the conclusion to having a submatrix that is only \textit{approximately} rank 1.
It is of course extremely interesting in both cases whether the bound of $\tilde{O}(\sqrt{n})$ can be improved further, ideally all the way to  $polylog(n)$.
In the Boolean setting, such a bound might prove the log rank conjecture,\footnote{We note a caveat that this depends on the notion of ``approximate'' used.  Gavinsky and Lovett~\cite{DBLP:conf/icalp/GavinskyL14} showed that to prove the log rank conjecture it suffices to find a in a rank $n$ Boolean matrix a rectangle of measure $\exp(-\polylog(n))$ that is \textit{nearly monochromatic} in the sense of having a $1-1/O(n)$ fraction of its entries equal. In this paper we are more concerned with rectangles whose distance to being rank one (or monochromatic) is some $\e>0$ that is only a small constant or $1/\polylog(n)$. \label{fn:approx-monochromatic}} while in our setting such a bound (assuming it extends to "pseudo matrices") would yield a quasipolynomial time algorithm for BSS, hence showing that $QMA(2) \subseteq EXP$. 
It can be shown that as stated, Theorem~\ref{thm:rank-one-reweighting} is tight.
However there are different notions of being "close to rank one" that  could be useful in both the log-rank and the quantum separability setting, for which there is hope to obtain substantially improved quantitative bounds. 
We discuss some of these conjectural directions in Section~\ref{sec:conclusions}.

\subsection{Overview of proof}
\label{sec:proof-overview}

In the rest of this technical overview, we give a proof sketch of Theorem~\ref{thm:rank-one-reweighting-dual} and then discuss how the proof can be "lifted" to hold in the setting of sum of square pseudo-distributions.
The condition that a matrix $A$ is of rank $n$ is the same as that $A = UV^\top$ where $U,V$ are two $n\times N$ matrices with columns $u_1,\ldots,u_N$ and $v_1,\ldots,v_N$ respectively (i.e., $A_{i,j}=\iprod{u_i,v_j}$ for all $i,j\in [N]$). 
We will restrict our attention to the case that all the columns of $U$ and $V$ are of unit norm. 
(This restriction is easy to lift and anyway holds automatically in our intended application.)
In this informal overview, we also restrict attention to the \textit{symmetric} case, in which $A=A^\top$ and can be written as $A=UU^\top$ and also assume that $U$ is \emph{isotropic}, in the sense that $\E_{i\in [N]} u_iu_i^\top = \tfrac{1}{n} \Id$.

Our inspiration is Lovett's result~\cite{DBLP:conf/stoc/Lovett14} which establishes a stronger conclusion for Boolean matrices. 
In particular, our proof follows Rothvo\ss's proof~\cite{DBLP:journals/corr/Rothvoss14a} of Lovett's theorem, though the non-Boolean setting does generate some non-trivial complications.
The $N\times N$ matrix $A$ satisfies that $A_{i,j} = \iprod{u_i,u_j}$.
An equivalent way to phrase our goal  is that we want to find a  subset $I\subseteq [N]$ over the indices such that:

\begin{description}
\item[\textbf{(i)}] $|I| \geq \exp(-\tilde{O}(\sqrt{n}))N$.
\item[\textbf{(ii)}] If $\lambda_1 \geq \lambda_2 \geq \cdots \lambda_n$ are the eigenvalues of $\E_{i\in I} u_iu_i^\top$ then $\epsilon^2 \lambda_1^2 \geq \sum_{j=2}^n \lambda_j^2$  
\end{description}

We will chose the set $I$ \emph{probabilistically} and show that \textbf{(i)} and \textbf{(ii)} above hold in \emph{expectation}. It is not hard to use standard concentration of measure bounds to then  deduce the desired result but we omit these calculations from this informal overview.

Our initial attempt for the choice of $I$ is simple, and is directly inspired by \cite{DBLP:journals/corr/Rothvoss14a}. We choose a random standard Gaussian vector $g\in N(0,\tfrac{1}{n}\Id)$ (i.e., for every $i$, $g_i$ is an independent standard Gaussian of mean zero and variance $1/n$). 
We then define $I_g = \{ i : \iprod{g,u_i} \geq \sqrt{k/n} \}$ where $k=\tilde{O}(\sqrt{n})$ is a parameter to be chosen later. 
Since $u_i$ is a unit vector,  $\iprod{g,u_i}$ is a  Gaussian of variance $1/n$, and so for every $i$, the probability that $i\in I_g$ is $\exp(-O(k))$ hence satisfying \textbf{(i)} in expectation.

The value $\lambda_1$ of $\E_{i\in I} u_iu_i^\top$ will be at least $\Omega(k/n)$ in expectation.
Indeed, we can see that the Gaussian vector $g$ that we choose (which satisfies $\norm{g}^2 = 1\pm o(1)$ with very high probability) will satisfy that 
$g^\top \left( \E_i u_i u_i^\top \right) g = \E_{i\in I_g} \iprod{u_i,g}^2 \geq k/n$
and hence in expectation the top eigenvalue of $\E_i u_iu_i^\top$ will be at least $(1-o(1))k/n$.

So, if we could only argue that in expectation it will hold that $\sum_{j=1}^n \lambda_j^2 \ll k^2/n^2 = \mathrm{polylog}(n)/n$ then  we'd be done.
Alas, this is not necessarily the case.
However, if this does fail, we can see that we have made progress, in the sense that by restricting to the indices in  $I$ we raised the Frobenius norm of $\E u_iu_i^\top$ from the previous value of $1/n$ (under the assumption that $U$ was isotropic) to $polylog(n)/n$. 
Our idea is to show that this holds in  general: we can select a Gaussian vector $g$ and define the set $I_g$ as above such that by restricting to the indices in $I_g$ we either get an approx rank one matrix or we increase  the Frobenius norm of our expectation matrix by at least an $(1+\epsilon)$ factor for an appropriately chosen $\epsilon>0$. 
Since the latter cannot happen more than $\log n/\epsilon$ times, the final set of indices still has measure $\exp(-\tilde{O}(\sqrt{n}))$. 

In further rounds, if our current set of indices is $I$ and the matrix (after subtracting from each vector $u_i$ its expectation) $U_I = \E_{i\in I} u_iu_i^\top = \sum_{j=1}^n \lambda_j v_jv_j^\top$ is not approximately rank one, then 
rather than choosing $g$ as a standard Gaussian, we choose it from the distribution $N(0,U_I)$ where we use $U_I$ as the covariance matrix.  
The expected norm of $g$ is simply $\Tr(U_I)$ which equals $1$.
For every  $i$, the random variable $\iprod{u_i,g}$ is a Gaussian with mean zero and variance $\sum_{j=1}^n \iprod{u_i,v_j}\lambda_j$. 
But for every $j$ in expectation over $i$, $\E \iprod{u_i,v_j}^2 = \lambda_j$ and so it turns out 
that we can assume that this random variable has variance $\sum \lambda_j^2 =  \norm{U_I}_F^2$. 

This means that if we choose $I' = \{ i\in I: \iprod{u_i,g} \geq \sqrt{k}\norm{U_I}_F \}$ we get a subset of $I$ with measure $\exp(-O(k))$. 
But now the new matrix $U_{I'} = \E_{i\in I'} u_iu_i^\top$ will have an eigenvalue of at least $k\norm{U_I}_F^2$ magnitude which is much larger than $\norm{U_I}_F$ since we chose $k\gg \sqrt{n}$. 
Hence $U_{I'}$ has significantly larger Frobenius norm than $U_I$. The above arguments can be made precise and yield a proof of Theorem~\ref{thm:rank-one-reweighting-dual} and thus also Theorem~\ref{thm:rank-one-reweighting}.

\subsection{Rectangle lemma for pseudo-distributions}

The above is sufficient to show that given $N\times n$ matrices $U=(u_1|\cdots|u_N)$ and $V=(v_1|\cdots|v_n)$ (which we view as inducing a distribution over rank one matrices by taking $u_iv_i^\top$ for a random $i$), we can condition on a not too unlikely event (of probability $\exp(-\tilde{O}(\sqrt{n}))$ to obtain that $\E u_iv_i^\top$ is roughly rank one. 
But in the sos setting we are \emph{not} given such matrices. 
Rather we have access to an object called a "pseudo-distribution" $\mu$ which we behaves to a certain extent as if it is such a distribution, but for which it is not actually the case.
In particular, we are not able to sample from $\mu$, or condition it on arbitrary events, but rather only compute $\E_\mu f(X)$ for polynomials $f$ of degree at most $\tilde{O}(\sqrt{n})$, and even these expectations are only ``pseudo expectations'' in the sense that they do not need to correspond to any actual probability distribution.

To lift the arguments above to the sos setting, we need to first show that if $\mu$ was an actual distribution, then we could perform all of the above operations using only access to $\tilde{O}(\sqrt{n})$ degree moments of $\mu$. 
Then we need to show that our \emph{analysis} can be captured by the degree $\tilde{O}(\sqrt{n})$ sos proof systems.
Both these steps, which are carried out in Sections~\ref{sec:fix-scalar} and ~\ref{sec:fix-vector} of this paper, are rather technical and non-trivial, and we do not describe them in this overview.

For starters, we need to move from \emph{conditioning} a probability distribution to \emph{reweighting} it. 
All of our  conditioning procedures above  had  the form of restricting to $i$'s such that $\xi(i) \geq \sqrt{k}$ where $\xi(i)$ was probabilistically chosen so that for every $i$ $\xi(i)$ is a  a mean zero and standard deviation one random variable satisfying $\Pr[ \xi(i) = \ell ] = \exp(-\Theta(\ell^2))$. 
We replace this conditioning by \emph{reweighting} the distribution $i$ with the function $\zeta(i)=\exp(\sqrt{k}\xi(i))$. 
Note that iterative conditioning based on functions $\xi_1,\ldots,\xi_t$ can be replaced with reweighting by the product function $\zeta_1,\ldots,\zeta_t$. 
We then show that these $\zeta_j$ functions can be approximated by  polynomials of $\tilde{O}(k)$ degree. 

The arguments above allow us to construct a rounding algorithm that at least makes sense syntactically, in the sense that it takes the $\tilde{O}(\sqrt{n})$ degrees moments of $\mu$ and produces a rank one matrix that is a candidate solution to the original matrix. 
To analyze this algorithm, we need to go carefully over our analysis before, and see that all the arguments used can be embedded in the sos proof system with relatively low degree.
Luckily we can rely on the recent  body of works that establishes a growing toolkit of techniques to show such embeddings~\cite{sos-lecture-notes}.

\section{Preliminaries}
\label{sec:preliminaries}

\Dnote{this section should consist of concise definitions of the concepts that are need to understand the technical sections.
We don't have to include motivation.
All motivation should have appeared in the previous sections.}

\Bnote{Made it shorter}

We use the standard $O(\cdot)$ and $\Omega(\cdot)$ notation to hide absolute multiplicative constants. 
We define $\Sn$ to be the $n-1$ dimensional unit sphere $\{ x\in\R^n : \sum_i |x_i|^2 =1 \}$.
For vectors $x \in \R^n$, we write $\Norm{x} = \sqrt{\sum_{i = 1}^n x_i^2}$ to denote the standard Euclidean norm. 
For matrices $A \in \R^{n \times n}$, we write $\Norm{A}$ to denote the spectral norm: $\max_{v: \Norm{v} = 1} |\langle v, Av \rangle |$ and $\Norm{A}_F$ to denote the Frobenius norm: $\sqrt{\sum_{i,j} A_{i,j}^2}.$

We use the following definitions related the sum of squares (sos) algorithm; see~\cite{sos-lecture-notes} for a more in-depth treatment.

\begin{definition} Let $n\in \N$ and $[x_1, x_2, \ldots, x_n]_d$ be the subspace of all $n$-variate real polynomials of degree at most $d$.
A \emph{degree-$d$ pseudo distribution $\mu$ over $\R^n$} is a finitely supported function from $\R^n$ to $\R$ such that if we define $\pE_\mu f = \sum_{x\in\mathrm{Supp}(\mu)}\mu(x)\cdot f(x)$ then $\pE_\mu 1 = 1$ and $\pE_\mu f^2 \geq 0$ for every $f \in \R[n]_{d/2}$.
We call $\pE_\mu f$ the \emph{pseudo-expectation} of $f$ with respect to $\mu$.
We will sometimes use the notation $\pE_{\mu(x)} f(x)$ to emphasize that we apply the pseudo expectation to the polynomial $f$ that is is taken with respect to the formal variables $x$. 

If $q \in \R[x_1, x_2, \ldots, x_n]_{d'}$, we say that a degree $d$ pseudo distribution $\mu$ \emph{satisfies} the constraint $\{ q \geq 0 \}$ if $\pE_\mu q \cdot f^2 \geq 0$ for every $f \in \R[n]_{(d-d')/2}$. 
We say that $\mu$ satisfies the constraint $\{ q = 0 \}$ if $\pE qf = 0$ for every $f\in \R[x_1, x_2, \ldots, x_n]_d$.
We say that $\mu(x)$ is a pseudo-distribution over the sphere or the unit ball if it satisfies $\{\norm{x}^2=1\}$ or $\{\norm{x}^2\le 1\}$.

If $\mu$ is a degree $d$ pseudo-distribution and $r\in \R[x_1, x_2,\ldots, x_n]_k$ a sum-of-squares polynomial with $k\leq d$, then the degree $d-k$ pseudo distribution $\mu' = r\cdot \mu$ is called  a \emph{degree-$k$ reweighting of $\mu$}. Note that $\mu'$ satisfies all constraints of degree at most $d-k$ that are satisfied by  $\mu$.
\end{definition}

\Dnote{I guess we are slightly cheating here.
  We are defining pseudo-distributions over the sphere.
  Later we use pseudo-distributions over the ball and products of spheres.} \Bnote{We could define pseudo-distributions for the ball in which case we would be fine in all cases, right?}

The sos algorithm is given as input a set of constraints $\cE$, a polynomial $q$,  and a parameter $d$, and runs in time $n^{O(d)}$ and  outputs the degree $d$ pseudo-distribution $\mu$ that satisfies all the constraints in $\cE$ maximizes $\pE_\mu q$ (see \cite{sos-lecture-notes}). In the special case of scalar real valued random variables, there's an actual distribution that agrees with a degree $d$ pseudo-distribution on all degree at most $d-1$ polynomials.

 \begin{fact} [See Corollary 6.14 in \cite{reznick2000some}, see also \cite{lasserre2015introduction}]
  Suppose $\mu$ is a pseudo-distribution on $\R$ of degree $d$.
  Then, there's an actual distribution $\mu'$ over $\R$ such that $\E_{\mu'}p = \pE_{\mu}p$ for every polynomial $p$ of degree at most $d-1$. \label{fact:univariate}
\end{fact}

\Pnote{we should import such facts to the lecture notes.}
% The above fact gives a useful tool to prove facts about pseudo-distributions - so long as they are facts about univariate distributions over the real line, we can prove the facts for actual distributions and then appeal to Fact \ref{fact:univariate} to obtain the same result for arbitrary pseudo-distributions of large enough degree.

\Pnote{Removed the section on basic facts etc and shortened the discussion on univariate fact above.}

\Bnote{I don't think we need to have these here - maybe in an appendix - by now we have repeated these in many papers. 
Maybe we should add a "toolkit" appendix to the sos lecture notes and cite it.}

\section{The Algorithm}
\label{sec:algorithm}

We now describe our algorithm, and show its analysis. 
A crucial tool for the analysis is the following  general \textit{structure theorem} on distributions over rank one matrices:

\begin{theorem}[Structure theorem for pseudo-distributions on rank one]\label{thm:general-structure-theorem}
Let $\e>0$, let $\mu$ be a pseudo-distribution over $\sphere \times \sphere$ of degree at least $k+2$, where $k= \sqrt{n} \log^C n/\epsilon^2$ for an absolute constant $C\ge 1$.
  Then, $\mu$ has a degree-$k$ reweighting $\mu'$ such that for $u_0 = \E_{\mu'(u,v)}u$ and $v_0=\E_{\mu'(u,v)}v$,
\[
\Norm{u_0v_0^\top -\pE_{\mu'(u,v)} uv^\top}_F \leq \epsilon \cdot \Norm{u_0v_0^\top}_F \mper
\]
Furthermore, we can find the reweighting polynomial $p=\mu'/\mu$ in time $2^{O(k)}$ and $p$ has only rational coefficients in the monomial basis with numerators and denominators of magnitude at most $2^{O(k)}$.
\end{theorem}

Theorem~\ref{thm:general-structure-theorem} is proven in Section~\ref{sec:structure-thm}. Our algorithm uses it as follows:

\newcommand{\algname}{4.1\xspace}

\begin{center}
\fbox{\begin{minipage}{6in} 
\begin{center}
\textbf{Algorithm \algname} 
\end{center}

\begin{description}
\item[Input:] Subspace $\cW \subseteq \R^{n^2}$ (in the form of a basis), and parameter $\epsilon>0$.

\item[Operation:] ~
\begin{enumerate}

\item Let $k=\sqrt{n} \log^C n/\epsilon^2$ be set as in the statement of  Theorem~\ref{thm:general-structure-theorem}.

\item Run the sum-of-squares algorithm to obtain a degree $k+2$ pseudo-distribution $\mu$ over pairs of vectors $(u,v) \in \R^{2n}$ that satisfies the constraint $u\transpose{v} \in \cW$ and $\norm{u}^2=\norm{v}^2=1$.
If no such pseudo distribution exists, output \texttt{FAIL}.

\item Use the procedure from Theorem~\ref{thm:general-structure-theorem} to find a degree $k$ reweighting $\mu'$ of $\mu$  such that 
$\norm{ \pE_{\mu'} u\transpose{v}   - u_0\transpose{v_0} }_F  \leq  \epsilon \norm{u_0\transpose{v_0}}_F$ where $u_0 = \pE_{\mu'} u$ and $v_0 = \pE_{\mu'} v$.

\item Output $u_0\transpose{v_0}$.

\end{enumerate}
\end{description}
\end{minipage}}
\end{center}

As discussed in Section~\ref{sec:rounding-from-rank-one}, the following theorem immediately implies our main result (Theorem~\ref{thm:main-thm}):

\begin{theorem}[Analysis of algorithm] \label{thm:alg-analysis}
Let $\epsilon>0$ and $\cW \subseteq \R^{n^2}$ be a linear subspace.
Then on input a basis for $\cW$, if there exists a nonzero rank one matrix $u\transpose{v}\in \cW$ then Algorithm~\algname will output a nonzero rank one matrix $L$ such that $\norm{\Pi_{\cW}L}_F^2 \geq (1-\epsilon^2)\norm{L}_F^2$ where $\Pi_{\cW}$ is the projector to $\cW$. 
\end{theorem}
\begin{proof}
Under the assumptions of the theorem, there exists a nonzero rank one matrix $uv^* \in \cW$ and by scaling we can assume $\norm{u}=\norm{v}=1$ and hence the degree $d$ SOS algorithm will return a pseudo-distribution $\mu$ satisfying these constraints for every $d$.
Since a reweighting $\mu'$ of a pseudo-distribution $\mu$ satisfies all constraints $\mu$ satisfies, we get that 
$\pE_{\mu'} uv^* \in \cW$. 
Hence the $\ell_2$ (i.e. Frobenius)  distance between the nonzero rank one $u_0v_0^*$ output by Algorithm \algname and the subspace $\cW$ will be at most $\epsilon \norm{u_0v_0^*}_F$ thus completing the proof.

\end{proof}
\begin{remark}[Imperfect Completeness] \label{rem:imperfect-completeness}
Note that the proof would have gone through even if the pseudo-distribution $\mu$ did not satisfy the condition that $u\transpose{v}\in \cW$ but merely that $\norm{\Pi_{\cW^\perp}u\transpose{v}} \ll \norm{u_0\transpose{v_0}}$ where $\Pi_V$ is the projector to a subspace $V$. 
The proof of Theorem~\ref{thm:general-structure-theorem} actually guarantees that $\norm{u_0\transpose{v_0}} \geq k/n$ which means that it
suffices that $\norm{\Pi_W u\transpose{v}}^2 \geq 1 - k^2/n^2$ hence implying that the proof works for the near perfect completeness case, as mentioned in Remark~\ref{rem:perfect-completeness}.
\end{remark}

% Next, we give a black-box reduction from the problem of finding a rank 1 matrix in subspaces of $\C^{n^2}.$

% \begin{corollary}
% Let $\epsilon>0$ and $\cW \subseteq \C^{n^2}$ be a linear subspace. 
% Then on input a basis for $\cW$, if there exists a nonzero rank one matrix $u\transpose{v}\in \cW$ then there is an algorithm that outputs a nonzero rank one matrix $L$ such that $\norm{\Pi_{\cW}L}_F^2 \geq (1-\epsilon^2)\norm{L}_F^2$ where $\Pi_{\cW}$ is the projector to $\cW$. 
% \end{corollary}

% \begin{proof}
% Write each complex vector $x$ as $x = x_r + j x_i$ where $j = \sqrt{-1}$ and $x_r$ and $x_i$ are the real and imaginary parts of $x$ respectively. Observe that if $x \in \sphere(\C)$ then vector $X$ obtained by concatenating $x_r$ and $x_i$ is contained in $\sphere(\R).$ Observe that every linear constraint over $\C$ on $x$ is equivalent to $2$ linear constraints on $X$ over $\R$. Thus, for any $\cW$, there's a subspace $\cZ \subseteq \R^{n^2}$ such that $x \in \cW$ iff $X \in \cZ$. 

% \end{proof}

%Statement and proof of structure theorem modulo the reweighting claims
%!TEX root = ../subexpalg.tex
\section{Structure Theorem}
\label{sec:structure-thm}

In this section, we prove that every pseudo-distribution over the unit sphere has a $\tilde O(\sqrt n)$-degree reweighting with second moment close to rank-1 in Frobenius norm. 
As discussed in Section~\ref{sec:intro:mono-to-reweight}, this theorem can be thought of as an approximate and robust  variant of Lovett's rectangle lemma \cite{DBLP:conf/stoc/Lovett14}.

\begin{theorem}[Structure theorem, real symmetric version]
  \label{thm:structure-reals}
  Let $\e>0$, let $\mu$ be a pseudo-distribution over $\sphere$ of degree at least $k+2$, where $k= \sqrt{n} (\log n)^C/\epsilon^2$ for an absolute constant $C\ge 1$.
  Then, $\mu$ has a degree-$k$ reweighting ("symmetric  rank 1 reweighting") $\mu'$ such that for $m = \E_{\mu'}x$,
  \[
    \Norm{\dyad{m} -\pE_{\mu'(x)} \dyad x}_F \leq \epsilon \cdot \Norm{\dyad m}_F \mper
  \]
Furthermore, we can find the reweighting polynomial $p=\mu'/\mu$ in time $2^{O(k)}$ and $p$ has only rational coefficients in the monomial basis with numerators and denominators of magnitude at most $2^{O(k)}$.
\end{theorem}
\begin{remark}
Our techniques extend to show similar structure theorem for pseudo-distributions over rank $r > 1$. For e.g., in Section \ref{sec:higher-rank-structure-theorem} of the Appendix, we give a higher-rank version of the structure theorem.
\end{remark}
The following more general version (see  Section~\ref{sec:app-structure-thm} for a proof) will be useful for the analysis of our algorithm from the previous section. We note that the previous theorem suffices for the symmetric analog of Algorithm~\algname.

\Bnote{Does this imply the rank r version and hence in particular also gives us the complex version?}
\Pnote{There are some complications with general $r$ - went back to the two dimensional version. Proof is the appendix.}

\begin{theorem}[2-Dimensional structure theorem] 
\label{thm:multi-dimensional-rank1}
Let $\e>0$, let $\mu$ be a pseudo-distribution over $(u_1,u_2) \sim (\sphere)^2$ of degree at least $k+2$, where $k= \sqrt{n} (\log n)^C/\epsilon^2$ for an absolute constant $C\ge 1$.
Then, $\mu$ has a degree-$k$ reweighting ("asymmetric rank 1 reweighting")  $\mu'$ such that for each $1 \leq j \leq 2$ 
\[
\Norm{\dyad{m_i} -\pE_{\mu'(u_i)} \dyad u_i}_F \leq \epsilon \cdot \Norm{\dyad{m_i}}_F \mcom
\]
where $m_i = \E_{\mu'(u_i)}u_i \mper$
Furthermore, we can find the reweighting polynomial $p=\mu'/\mu$ in time $2^{O(k)}$ and $p$ has only rational coefficients in the monomial basis with numerators and denominators of magnitude at most $2^{O(k)}$.
\end{theorem}

Theorem~\ref{thm:multi-dimensional-rank1} directly implies Theorem~\ref{thm:general-structure-theorem}. Indeed, if we write $u = u_0 + u'$ and $v=v_0+v'$ where $u',v'$ are mean zero random variables, then we see that 
\[
\E (u_0+u')(v_0+v')^\top = u_0v_0^\top + \E u'v'^\top 
\]
but $\norm{\E u'v'^\top}^2 \leq \norm{\E u'u'^\top} \norm{\E v'v'^\top}$.

We present the proof of Theorem \ref{thm:multi-dimensional-rank1} which is similar to that of Theorem \ref{thm:structure-reals} in Section~\ref{sec:app-structure-thm} of the Appendix.

\subsection{Reweighting Schemes}
The proof of Theorem~\ref{thm:structure-reals} is based on the following general results about existence of low-degree SoS reweighting schemes. We prove these results in the following sections. 

The first lemma shows that there's a low-degree reweighting for any pseudo-distribution over an interval in $\R$ such that the resulting distribution is concentrated around the old standard deviation. 

\begin{lemma}[Scalar Fixing Reweighting: Fixing a scalar around its standard deviation] \label{lem:fixing-scalar}
Let $\mu$ be a distribution over $\R$ satisfying $\{x \leq n\}$ and $\pE_{\mu}x^2 \geq 1.$  Then, for some absolute constant $C > 0$, there exists a reweighting ("scalar fixing reweighting") $\mu'$ of $\mu$ of degree $k = Cd\log{(n)}/\epsilon^2$ satisfying $\pE_{\mu'} (x-m)^d \leq \epsilon^d m^d$ for some $m$ satisfying $|m| \geq 1.$ Moreover, the conclusions hold even for pseudo-distributions $\mu$ of degree at least $d+k.$
\end{lemma}

We consider this ``fixing'' the distribution, since if $\mu$ and $\mu'$ were actual distributions, the conclusion of Lemma~\ref{lem:fixing-scalar} would imply that $\Pr_{\mu'}[ |x-m| \geq 2\epsilon m ] \leq 2^{-k}$.
%Since the proof of Lemma~\ref{lem:fixing-scalar} is rather technical, we defer it to  Appendix~\ref{sec:fix-scalar}.

Next, we show that for pseudo-distribution of degree at least $O(d)$ over the $d$-dimensional unit ball have $O(d)$-degree reweightings such that the resulting distribution is concentrated around a single vector. This result is  related to previous results on using high-degree sum-of-squares relaxations for optimizing general polynomials over the unit sphere \cite{DBLP:journals/corr/abs-1210-5048}. % Doherty-Wehner
However, the previously known bounds are not strong enough for our purposes.

\begin{lemma}[Subspace Fixing Reweighting: Fixing a distribution in a subspace]
 \label{lem:fixing-subspace}
 For every $C\geq 1$ and $\delta > 0$, there is some $C'$, such that if $\mu$ is a distribution over the unit ball $\{ x : \norm{x} \leq 1 \}$ of $\R^d$ such that  $\E_{\mu}\norm{x}^2 \ge d^{-C}$ then there is a degree $k= \frac{d}{\delta} \cdot (\log d)^{C'}$ reweighting $\mu'$ of $\mu$  such that
  \[
    \Norm{\E_{\mu'(x)}x}^2 \ge (1-\delta) \E_{\mu(x)}\Norm{x}^2  \mper
  \]
 Further, the reweighting polynomial $p=\mu'/\mu$ can be found in time $2^{O(k)}$, has all coefficients upper bounded by $2^{O(k)}$ in the monomial basis, and satisfies $p(x) \leq k^{O(k)} \norm{x}^{k}$.
The result extends to pseudo-distributions $\mu$ of degree at least $d = k+2$, in which case, the reweighted pseudo-distribution $\mu'$ is of degree $d - k.$ 
\end{lemma}

\Bnote{should we state this with $1-\epsilon$ instead of $0.99$? seems like a general statement that would be useful elsehwhere.}
% \begin{lemma}[Fixing a distribution in a subspace]
% \label{lem:fixing-subspace}
% Let $\mu$ be a distribution over $\sphere$ such that $\pE_{\mu}[\norm{x}^2] \le d^{O(1)}$ and $\pE_{\mu}[\norm{x}^2] \geq 1$. Then, there is an absolute constant $C$ such that for $k = d \log^C{(n)}$, there exists a degree-$k$ reweighting $\mu'$ of $\mu$ such that
% \[
% \Norm{\pE_{\mu'}x} \ge 0.5 \mper
% \]
% Further, the reweighting polynomial $p=\mu'/\mu$ can be found in time $2^{O(k)}$, has all coefficients upper bounded by $2^{k}$ in the monomial basis, and satisfies $p(x) \leq k^k \norm{x}^{k}$. Moreover, the conclusions above hold even if $\mu$ is a pseudo-distribution of degree at least $k+2$.
% \end{lemma}

We prove Lemma~\ref{lem:fixing-subspace} in Section~\ref{sec:fix-vector}.

\subsection{Proof of Structure Theorem}
We now prove Theorem~\ref{thm:structure-reals} using Lemmas~\ref{lem:fixing-scalar}, \ref{lem:fixing-subspace}.
% and Corollary~\ref{cor:making-progress}. 

The key tool will be the following direct corollary of Lemma \ref{lem:fixing-subspace} that allows us to argue that we make progress in every iteration of the Algorithm.

\begin{corollary} \label{cor:making-progress}
Fix $\epsilon > 0.$ Let $\mu$ be a distribution on $\sphere$ such that $\epsilon \Norm{ \E_{\mu} x }^2 < \Norm{\E_{\mu} \dyad{x} - \dyad{\pE_{\mu}x}}_F.$ Then, for an absolute constant $C$, there's a SoS polynomial $p$ of degree $k = \frac{\sqrt{n}}{\epsilon} \log^C{(n)}$ such that the reweighting $\mu' = \mu \cdot p$ satisfies $\Norm{\E_{\mu'}x}^2> \max \{ (1+\epsilon/4) \Norm{\E_{\mu}x}^2, \frac{1}{4\sqrt{n}} \}.$ The result extends to pseudo-distributions $\mu$ of degree at least $d = k+2$, in which case, the reweighted pseudo-distribution $\mu'$ is of degree $d - k.$ 
\end{corollary}

\begin{proof}
The key observation is that under the hypothesis of the corollary, there's a subspace $S$ of dimension at most $\sqrt{n} + 3$ such that either $\pE_{\mu}{\Norm{\Pi_S x}^2} > \frac{1}{\sqrt{n}}$ or $\pE_{\mu}[ \Norm{\Pi_S x}^2 ] \geq (1+\epsilon) \Norm{ \pE_{\mu} x}^2$, where $\Pi_S$ is the projector to the subspace $S.$ Applying lemma \ref{lem:fixing-subspace} to the random  variable $\Pi_S x$ using a $k/\delta \poly \log{k}$ degree reweighting (with $k = \sqrt{n} +3$ and $\delta = \min \{ \frac{1}{2}, \epsilon/2\}$ ) then yields a distribution $\mu'$ such that $\Norm{\pE_{\mu'} x}^2 > \Norm{\pE_{\mu'}\Pi_S x}^2 > (1-\delta)\pE_{\mu}{\Norm{\Pi_S x}^2}$ giving us the corollary. We now prove the observation in the two claims below to complete the proof.

For any subspace $S$, we write $\Pi_S$ for the associated projector matrix.
\begin{claim}
Let $\mu$ be a pseudo-distribution on $\R^n$ of degree at least $\lceil \sqrt{n} \rceil +3$. There exists a subspace $S$ of dimension at most $\lceil \sqrt{n} \rceil +1$ such that $$\pE_{\mu} \Norm{  \Pi_S x }^2 \geq \Norm{ \pE_{\mu} \dyad{x}}_F.$$
\label{claim:trace-equal-Frob}
\end{claim}

\begin{proof}
Let $\lambda_1 \geq \lambda_2 \geq \cdots \geq \lambda_n \geq 0$ be the eigenvalues of $\pE_{\mu} \dyad{x}.$ For $\ell = \lceil \sqrt{n} \rceil + 1$, let $S$ be the subspace spanned by the eigenvectors corresponding to eigenvalues $\lambda_1, \lambda_2, \ldots, \lambda_\ell.$ For any $x \in \R^n$, let $x_S$ denote the projection of $x$ to the subspace $S$.

Then, observe that 
\begin{equation}
\sum_{i = \ell+1} \lambda_i^2 \leq n \lambda_{\ell+1}^2 \leq \frac{n}{\ell(\ell-1)} \sum_{i \neq j} \lambda_i \lambda_j. \label{eq:frob-error-bound}
\end{equation}

Thus, for $x_S = \Pi_S x$,
\begin{align}
\left(\pE_{\mu} \norm{x_S}^2\right)^2 &=  \left(\sum_{i = 1}^{\ell} \lambda_i\right)^2 \geq \sum_{i} \lambda_i^2 + \sum_{i \neq j \in [\ell]} \lambda_i \lambda_j \notag \\
&\geq \sum_{i\leq \ell} \lambda_i^2 + \frac{\ell(\ell-1)}{n} \sum_{i = \ell+1}^n \lambda_i^2 \geq \sum_{i = 1}^{n} \lambda_i^2 \notag\\
&= \Norm{\pE_{\mu}\dyad{x}}^2_F. \label{eq:growth-propto-frobenius}
\end{align}
\end{proof}

\begin{claim}
Let $\mu$ be a pseudo-distribution on $\R^n$ of degree at least $\lceil \sqrt{n} \rceil +4$. Suppose $\epsilon \Norm{ \E_{\mu} x }^2 < \Norm{\E_{\mu} \dyad{x} - \dyad{\pE_{\mu}x}}_F.$ Then, there's a subspace $S$ of dimension at most $\lceil \sqrt{n} \rceil +2$ such that $\pE_{\mu}\Norm{\Pi_S x}^2 \geq (1+\epsilon) \Norm{\pE_{\mu}x}^2.$
\end{claim} 

\begin{proof}
Observe that $\pE_{\mu} \dyad{(x-\pE_{\mu}x)} = \pE_{\mu} \dyad{x} - \dyad{(\pE_{\mu}x)}.$ Thus, in particular, the hypothesis implies that $\Norm{\pE_{\mu} \dyad{(x-\pE_{\mu}x)}}_F\geq \epsilon \Norm{\pE_{\mu}x}^2.$

Apply Claim \ref{claim:trace-equal-Frob} to the random variable $(x-\pE_{\mu}x)$ and obtain a subspace $S$ of dimension at most $\lceil \sqrt{n} \rceil + 1$ such that 
\begin{equation} \label{eq:good-subspace} \pE_{\mu}  \Norm{\Pi_S (x-\pE_{\mu}x)}^2 \geq \Norm{ \pE_{\mu} \dyad{(x-\pE_{\mu}x)}}_F \geq \epsilon \Norm{\pE_{\mu}x}^2. \end{equation}

Let $S'$ be the subspace spanned by the direction of $\pE_{\mu}x$ and $S$. Then, $S'$ has dimension at most $\lceil \sqrt{n} \rceil + 2.$ We claim that $S'$ satisfies the requirements of the claim. To verify this, we estimate:

\begin{align*}
\pE_{\mu} \Norm{ \Pi_{S'} x}^2 &= \pE_{\mu} \Norm{ \Pi_{S'} (x-\pE_{\mu}x)}^2 + \pE_{\mu} \Norm{ \Pi_{S'} \pE_{\mu}x}^2 + 2 \pE_{\mu} \langle  \Pi_{S'} (x-\pE_{\mu}x), \Pi_{S'} \pE_{\mu}x \rangle\\
&\geq \epsilon \Norm{\pE_{\mu}x}^2 + \Norm{\pE_{\mu}x}^2 + 0\\
&\geq (1+\epsilon) \Norm{\pE_{\mu}x}^2,
\end{align*}
where, in the second line, we used the lower bound on the first term computed above in \eqref{eq:good-subspace}, the observation that $S'$ includes the direction of $\pE_{\mu} x$ so $ \Pi_{S'} \pE_{\mu} x = \pE_{\mu} x$ for the second term and the observation that $\pE_{\mu} \Pi_{S'} x - \pE_{\mu} x = 0$ for computing the third term.
\end{proof}

\end{proof}

\begin{proof}[Proof of Theorem~\ref{thm:structure-reals}]
Our proof of the structure theorem is algorithmic and uses Corollary \ref{cor:making-progress} repeatedly. We describe the procedure below and then analyze it.

\newcommand{\algnametwo}{5.2\xspace}

\begin{center}
\fbox{\begin{minipage}{6in} 
\begin{center}
\textbf{Algorithm \algnametwo} 
\end{center}

\begin{description}
\item[Input:] $d \in \N$, a pseudo-distribution $\mu$ of degree $d$, and parameter $\epsilon>0$.
\item[Output:] A pseudo-distribution $\mu'$ of degree at least $2$ such that \begin{equation} \label{eq:stopping-condition}\Norm{\pE_{\mu'}\dyad{x} - \dyad{(\pE_{\mu'}x)})}_F \leq \epsilon \Norm{\dyad{(\pE_{\mu'}x)}}_F = \epsilon\Norm{\pE_{\mu'}x}^2.\end{equation}
\item[Operation:] ~
\begin{enumerate}
\item Set $\mu' \leftarrow \mu.$
\item If $\mu'$ satisfies \eqref{eq:stopping-condition}, output $\mu'$ and halt. 
\item Otherwise, apply the reweighting from Corollary \ref{cor:making-progress} to $\mu'$ and obtain a pseudo-distribution $\mu''.$
\item Set $\mu' \leftarrow \mu''$ and go to Step 2.

\end{enumerate}
\end{description}
\end{minipage}}
\end{center}

% Our proof is based on taking iterative reweightings of the pseudo-distribution $\mu$. We start from $\mu_0 = \mu$ with covariance $M_0 = \pE_{\mu}[x \transpose{x}]$ and construct $\mu_1, \mu_2, \ldots, \mu_q$ in the course of our construction with means $m_t = \pE_{\mu_t}x$. 
Our main claim is that if the pseudo-distribution that we begin with has degree $d > O(1/\epsilon^2) \sqrt{n} \log^{C+1}{(n)}$ then the procedure above terminates pseudo-distribution $\mu'$ of degree at least 2 as required. Let $\mu = \mu_0, \mu_1, \ldots, \mu_T$ be the sequence of pseudo-distributions constructed when applying the procedure above with the final pseudo-distribution being $\mu_T.$
%Let $x_{-t}$ defined by $x - \frac{\langle x, m_t \rangle}{\norm{m_t}^2} m_t$ if $m_t \neq 0$ and $x$ otherwise, denote the component of $x$ orthogonal to $m_t$ and set $C_t = \pE_{\mu_i} x_{-t} \transpose{x_{-t}}$ to be the covariance matrix of $x_{-t}$.

% We first describe the first step in the process above - where we obtain $\mu_1$ from $\mu_0$. 
% Let $\lambda_1 = \lambda_1(C_i) \geq \lambda_2 = \lambda_2(C_i) \geq \ldots \geq \lambda_n = \lambda_n(C_i) \geq 0$ denote the eigenvalues of $C_i$. 
% Since $\mu$ is a distribution on $\sphere$, $\sum_{j = 1}^n \lambda_j = \Tr(C) = \pE_{\mu} \norm{x}^2 = 1.$ Thus if we let $\ell = \lceil 2\sqrt{n} \rceil$, then since the sum of the first $\ell$ eigenvalues is larger than any block of $\ell$ eigenvalues, we get that $\lceil \tfrac{n}{\ell} \rceil \sum_{j=1}^\ell \lambda_j \geq 1$ or 
% $
% \pE_{\mu_0} \norm{x_S}^2 = \sum_{j=1}^\ell \lambda_j \geq 2/\sqrt{n}
% $
% where $S$ is the $\ell$-dimensional subspace spanned by the eigenvectors of $C_0$ corresponding to $\lambda_1,\ldots,\lambda_\ell$.

% Applying Lemma~\ref{lem:fixing-subspace} to the distribution of $x_S$, we obtain a reweighting $\mu_1$ of $\mu$ such that 
% \begin{equation}
% \Norm{\pE_{\mu_1}x}^2 \geq \Norm{\pE_{\mu_1}x_S}^2 \geq 0.5 \norm{\pE_{\mu_0}x \transpose{x}} \geq \frac{1}{\sqrt{n}}. \label{eq:norm-step-0}
% \end{equation}
If for some $i \geq 1$, $\mu_{i-1}$ doesn't satisfy the requirements of the structure theorem, then $\mu_{i-1}$ satisfies the hypothesis for Corollary \ref{cor:making-progress}. We then set $\mu_{i}$ to be the reweighting of $\mu_{i-1}$ as given by Corollary \ref{cor:making-progress} and obtain that $\Norm{\pE_{\mu_i}x}^2 > (1+\epsilon/4) \Norm{\pE_{\mu_{i-1}}x}^2.$ 

Further, observe that $\Norm{\pE_{\mu_1}x}^2 \geq \frac{1}{\sqrt{n}}.$ Thus, in every iteration $i > 1$ such that the stopping condition \eqref{eq:stopping-condition} is not reached, Corollary \ref{cor:making-progress} guarantees that $\Norm{\pE_{\mu_i}x}^2$ grows at least by a multiplicative $(1+\epsilon/4)$ factor. Thus, after $O(1/\epsilon) \log{(n)}$ iterations, the procedure must halt given that the degree $d$ is large enough. Since every iteration requires a reweighting of degree $\frac{ \sqrt{n}}{\epsilon} \log^C{(n)}$ via Corollary \ref{cor:making-progress}, the final $\mu_T$ has degree at least $d-O(1/\epsilon^2) \sqrt{n}\log^{C+1}{(n)} .$ Thus, using $d = O(1/\epsilon^2) \sqrt{n} \log^{C+1}{(n)} $ suffices.

\end{proof}

%!TEX root = ../subexpalg.tex

\section{Fixing scalar-valued random variables}
\label{sec:fix-scalar}

In this section, we prove Lemma~\ref{lem:fixing-scalar}. We begin by restating it.

\begin{lemma}[Scalar Fixing reweighting] \label{lem:fixing-scalar-restated}
Let $\mu$ be a distribution over $\R$ satisfying $\{x \leq n\}$ and $\pE_{\mu}x^2 \geq 1.$  Then, for some absolute constant $C$, there exists a reweighting ("scalar fixing reweighting") $\mu'$ of $\mu$ of degree $k = Cd\log{(n)}/\epsilon^2$ satisfying $\pE_{\mu'} (x-m)^d \leq \epsilon^d m^d$ for some $m$ satisfying $|m| \geq 1.$ Moreover, the conclusions hold even for pseudo-distributions $\mu$ of degree at least $d+k.$
\end{lemma}

It is instructive to derive intuition from a \emph{conditioning} version of the lemma above for actual probability distributions. Given a random variable $x$ with distribution $\mu$ that has standard deviation $1$ and is bounded in $[-n,n]$, we know that with probability at least $\Theta(\frac{\delta}{n^2})$ that $x^2 \geq 1-\delta.$ As a result, the probability of at least one of $x \geq 1-\Theta(\delta)$ or $x \leq - (1-\Theta(\delta))$, say the former, is also at least $\Theta(\frac{\delta}{n^2}).$ Next, we  partition $[1-\delta,n]$ into $O(\log{(n)})$ intervals with end points differing by a multiplicative factor of, say $1.1$. Then, from the above calculation, there's an interval in this partition such that $x$ is contained in it with probability at least $\Theta(\frac{\delta}{n^2 \log{(n)}}).$ Thus, if we condition on $x$ lying in the above chosen interval to obtain $\mu'$, then $KL(\mu || \mu') \leq O(\log{(n)} + \log{(1/\delta)}).$

Our plan is to roughly implement the above conditioning argument for pseudo-distributions. This demands that instead of conditioning, we use reweightings by low-degree SoS polynomials and that further, all our arguments hold for low-degree pseudo-distributions with degree roughly matching the KL-divergence bound above. 

We will use a general trick in order to aid us in this task. Specifically, Fact \ref{fact:univariate} says that any statement about expectations of degree $d$ polynomials apply to all pseudo-distributions over $\R$ with degree at least $d+1$. We will rely on this fact heavily in what follows.
We begin with a simple claim that we will use repeatedly in what follows. 

\begin{lemma}
Let $\mu$ be a pseudo-distribution of a random variable $x$ over $\R$. Let $\mu'$ be obtained by reweighting $\mu$ using $x^{2\ell}$ for some $\ell \in \N.$ Then, $\pE_{\mu'} x^2 \geq \pE_{\mu} x^2$ so long as $\mu$ has degree at least $2\ell+3.$ \label{lem:monotonicity-means}
\end{lemma}
\begin{proof}
Since the claim is about a pseudo-distribution over $\R$, we can appeal to Fact~\ref{fact:univariate} and thus it suffices to show the result for arbitrary actual probability distributions over $\R.$ An application of Holder's inequality shows that $\pE_{\mu} x^2 \pE_{\mu} x^{2\ell} \leq \pE_{\mu}x^{2\ell+2}.$ Rearranging yields, $\pE_{\mu'} x^2 = \pE_{\mu} x^{2\ell+2} /\pE_{\mu} x^{2\ell} \geq \pE_{\mu} x^2.$ 
\end{proof}

Our main technical tool is the following lemma that shows that if we take two reweightings $\mu_1, \mu_2$ of a distribution $\mu = \mu_0$ associated with a random variable $x$ over $\R$ such that the means $\E_{\mu_i} x^2$ remain roughly the same, then, under $\mu_1$, $x^2$ is concentrated around its mean under $\mu_1$.  Notice that this is a statement about expectations of polynomials over  \emph{arbitrary} distributions over the real line and will thus immediately extend to pseudo-distributions using Fact~\ref{fact:univariate} as described above.

\begin{lemma}
Fix $\epsilon > 0$ and $d\in \N$. Let $\mu = \mu_0$ be a distribution over $\R$ satisfying $\pE_{\mu(x)} x^{2} = 1.$ Then, for any $k > 4 + \frac{2d \log{(1/\epsilon)}}{\epsilon} ,$ and for successive reweightings $\mu_1 = \mu \cdot p_1$ and $\mu_2 = \mu_1 \cdot p_2$ of $\mu$ using the polynomials $p_1 = \frac{x^{2k}}{\pE_{\mu_0} x^{2k}}$ and $p_2 = \frac{x^{2k}}{\pE_{\mu_1}x^{2k}}$ respectively, at least one of the following three consequences holds:
\begin{enumerate}
\item $\pE_{\mu_1} x^2 > (1+\epsilon) \cdot \pE_{\mu} x^2.$ 
\item $\pE_{\mu_2} x^2 >  (1+\epsilon) \cdot \pE_{\mu_2} x^2.$
\item $\pE_{\mu_1} (x^2-m)^{2d} \leq 3\epsilon^{2d} m^{2d}$ for $m = \pE_{\mu_1} x^2.$
\end{enumerate}
Moreover, the claim holds also for pseudo-distributions $\mu$ of degree at least $5k$. \label{lem:unconstrained-scalar-reweighting}
\end{lemma}
\begin{remark}
Observe the three statements above are claims about (pseudo-)expectations of degree at most $4k+2$ polynomials under $\mu_0 = \mu$. Specifically, the three conditions have the following equivalent form:
\begin{enumerate}
\item \[
\pE_{\mu} x^{2d+2} > (1+\epsilon) \pE_{\mu} x^{2d} \E_{\mu} x^2.
\]
\item
\[
\pE_{\mu} x^{4d+2} > (1+\epsilon) \pE_{\mu} x^{4d} \E_{\mu} x^2.
\]
\item
\[\pE_{\mu} x^{2d}\left(x^2-\frac{\pE_{\mu}x^{2d+2}}{\pE_{\mu} x^{2d}}\right)^{2d} \leq 3\epsilon^{2d} \left(\frac{\pE_{\mu}x^{2d+2}}{\pE_{\mu} x^{2d}}\right)^{2d} \pE_{\mu} x^{2d}.
\]
\end{enumerate}
\end{remark}

\begin{proof}
We prove the statement for actual distributions over the real line - notice that this allows us to make arguments that involve probabilities of various events. We then appeal to Fact \ref{fact:univariate} as discussed above to get the statement for pseudo-distributions.
Observe that by an application of Holder's inequality, $(\E_{\mu} x^{2k})^{2k} \geq \left(\E_{\mu} x^{2} \right)^{4k-2} \geq 1$. A similar argument shows that $\E_{\mu_1} x^{2k} = \frac{\E_{\mu} x^{4k}}{\E_{\mu} x^{2k}} \geq 1.$ Thus, $p_1, p_2$ are well-defined. %Lemma~\ref{lem:monotonicity-means} implies that the means under $\mu_i$ grow monotonically.
Now, suppose 1) and 2) do not hold. We write 
\begin{equation}
\E_{\mu_1} \left(\frac{x^2}{m} -1\right)^{2d} \leq 16^{2d} \epsilon^{2d} + \E_{\mu_1}  \left(\frac{x^2}{m} -1\right)^{2d} \1\left(\frac{x^2}{m} > (1+\epsilon)^4\right) + \E_{\mu_1}  \left(\frac{x^2}{m} -1\right)^{2d} \1\left(\frac{x^2}{m} < (1-\epsilon)^4\right). \label{eq:expectation-split-three}
\end{equation}

We bound the second term on the right hand side next.  When 1) and 2) do not hold, observe that $\E_{\mu_1} x^2, \E_{\mu_2} x^2 \in (1 \pm \epsilon) \E_{\mu} x^2.$  Using Lemma \ref{lem:monotonicity-means} and the form of $p_2$, this yields that for every $\ell \leq k$, $\E_{\mu_1} x^{2\ell+2} \leq (1+\epsilon) \E_{\mu_1} x^{2\ell} \E_{\mu_1} x^2.$ Repeatedly applying this inequality yields:
\begin{equation}
\E_{\mu_1} x^{2k+2} \leq (1+\epsilon)^{k} \left(\E_{\mu_1}x^2\right)^{k+1}. \label{eq:hypercontractivity}
\end{equation}

Using Markov's inequality along with \eqref{eq:hypercontractivity} yields:
\begin{equation}
\Pr_{\mu_1} [ \frac{x^2}{m} \geq q] \leq \left(\frac{1+\epsilon}{q} \right)^{k}. \label{eq:tail-bound-from-hypercontractivity}
\end{equation} 
Let $y = \frac{x^2}{m} -1.$ Then, we have:

\begin{align}
\E_{\mu_1} \left(\frac{x^2}{m} -1\right)^{2d} \1\left(\frac{x^2}{m} > (1+\epsilon)^4\right) &\leq \int_{y \geq 1+(1+\epsilon)^4} y^{2d} \mu(y) dy \notag \\
&\leq \int_{y \geq 1+(1+\epsilon)^4} y^{2d} \left(\frac{1+\epsilon}{1+y}\right)^{k} dy \notag \\
&\leq \int_{y \geq 1+(1+\epsilon)^4 } \frac{y^{2d}}{(1+y)^{k/2}} \left( \frac{1+\epsilon}{\sqrt{1+y}} \right)^k dy \notag \\
&\leq \int_{y \geq 1+(1+\epsilon)^4 } \frac{y^{2d}}{(1+y)^{k/2}} dy\notag \\
&= \int_{y \geq 1+(1+\epsilon)^4 } \frac{1}{(1+y)^{2}} \cdot \frac{y^{2d}}{(1+y)^{k/2 - 2}} dy \label{eq:final-bound}\\
\end{align}

Now, since $k > 4 + \frac{2d}{\epsilon}$, $y \leq \epsilon \frac{k-4}{2d} y \leq \epsilon(1+y)^{\frac{k-4}{2d}}.$ And thus, $\frac{y^{2d}}{(1+y)^{k/2 - 2}} \leq \epsilon^{2d}$ for every $y \geq 1+(1+\epsilon)^4.$ Thus, the expression in \eqref{eq:final-bound} is upper bounded by $\epsilon^{2d} \int_{y \geq 1} \frac{1}{(1+y)^2} \leq \epsilon^{2d}.$ This shows that the second term in  \eqref{eq:expectation-split-three} is upper bounded by $\epsilon^{2d}.$ 

Analyzing the third term in \eqref{eq:expectation-split-three} is easy: we have: $\Pr_{\mu_1} [x^2 < q] \leq \Pr_{\mu_0}[ x^2 < q] \cdot \max_{x^2 < q} p_1(x) \leq q^{2k}.$ For any $q < 1-\epsilon$, this quantity is at most $\epsilon^{2d}$ if $k > \frac{2d}{\epsilon} \log{(1/\epsilon)}$ and as a result, the third term in \eqref{eq:expectation-split-three} is at most $\epsilon^{2d}.$

This completes the proof.

% Let 
% \[
% x_{small} = \begin{cases} x & \text{ if } x \leq (1-\epsilon)m\\
% 0 & \text{ otherwise,}
% \end{cases}
% \]
% and
% \[
% x_{big} = \begin{cases} x & \text{ if } x > (1+\epsilon) m\\
% 0 & \text{ otherwise.}
% \end{cases}
% \]

% Then, 
% \[
% \E_{\mu} x^{2k}(x^2-m)^{2d} \leq \E_{\mu} x^{2k} \left( (x^2_{small} - m)^d + (x_{big}-m)^{2d} \right) + \epsilon^{2d} m^{2d} \E_{\mu} x^{2k}\mper
% \]

% Now, we claim that $\E_{\mu} x^{2k} (x_{big}-m)^{2d} \leq \epsilon^{2d} m^{2d} \E_{\mu} x^{2k}.$ To verify this, assume otherwise. Then, notice that the function $\frac{x^{2d}}{m^{2d}(1+\epsilon)^{2d}} - \left(\frac{(x-m)}{2m\epsilon} \right)^{2d} \geq 0 $ had a non-negative expectation under $\mu \cdot x^{2k}$. This is because the function is pointwise non-negative when $x \leq (1+\epsilon)m$, and non-negative in expectation conditioned on $x \geq (1+\epsilon)m$. As a result, $\E_{\mu} x^{2d} \geq (1+\epsilon)^{2d} m^{2d} \E_{\mu}x^{2k}.$

% A similar argument can be used to show that $\E_{\mu} x^{2k} (x_{small} -m)^{2d} \leq  \epsilon^{2d} m^{2d} \E_{\mu} x^{2k}.$
\end{proof}

\begin{remark}
It is important to note that even though \emph{our arguments} in the proof above require higher degree polynomials ($> 5k$) - such as when we apply Holder's inequality - the statements themselves are about non-negativity of polynomials of degree at most $4k+2$. Thus, an application of Fact~\ref{fact:univariate} shows that these non-negativity statements, when true, hold for any pseudo-distribution of degree $\geq 4k+3$. In particular, in situations as in the proof above, we do not have to be judicious in the use of the degree. 
\end{remark}

\begin{lemma}
Let $\mu$ be a distribution over $\R$ satisfying $\{x^2 \leq n\}$ and $\pE_{\mu}x^2 \geq 1.$ For some absolute constant $C$, there's a reweighting $\mu'$ of $\mu$ of degree $k = C d \log{(n)}/\epsilon^2$ such that $\pE_{\mu'}(x^2-m)^{2d} \leq \epsilon^{2d} m^{2d}$ for $m \geq 1.$ Moreover, the conclusions hold even for pseudo-distributions of degree at least $d+ k.$ \label{lem:positive-scalar-reweighting}
\end{lemma}

\begin{proof}
The statement is about a pseudo-distribution $\mu$ that is subjected to some constraints - we cannot now apply Fact \ref{fact:univariate} directly. So instead 1) we prove a claim about actual distributions that are unconstrained, i.e. over the reals 2) apply Fact \ref{fact:univariate} to obtain the same claim for pseudo-distributions 3) Show that the claim implies the conclusion of the lemma for \emph{constrained} pseudo-distributions. 

Formally, we take a sequence of reweightings $\mu_0 = \mu, \mu_1, \mu_2,\dots, \mu_r$ of $\mu$ such that $\mu_{i}/\mu_{i-1} = \frac{x^{2k}}{\pE_{\mu_{i-1}} x^{2k}}$ for each $r \geq i \geq 1.$ Observe that  Lemma~\ref{lem:monotonicity-means} implies that means of $x^2$ under $\mu_i$ grow monotonically and thus are all at least $1.$

We then apply Lemma~\ref{lem:unconstrained-scalar-reweighting} to every 3-tuple $\mu_{i-1}, \mu_i, \mu_{i+1}$ for $1 \leq i \leq r-1.$ If conclusion 3) from the statement of Lemma~\ref{lem:unconstrained-scalar-reweighting} does not hold, then, then in every consecutive triple of reweightings $\mu_{i-1}, \mu_i, \mu_{i+1}$ as above, at least one of the consecutive pairs has a multiplicative gap of $(1+\epsilon)$ in the means of $x^2.$  

We now come to the only place we use the constraints - since $\mu$ satisfies $\{X^2 \leq n \}$, we observe that for any positive polynomial reweighting of degree $d$, $\pE_{\mu'}[x^2] \leq n$ whenever $\mu'$ is of degree at least $d+2$. This follows from an application of Holder's inequality for pseudo-distributions. Thus, the number of successive triples of reweightings above that do not satisfy the condition 3) of Lemma is at most $O(\log{(n)}/\epsilon)$. 

Thus, if we choose $r = O(\log{(n)}/\epsilon)$, there's a consecutive triple of reweightings of $\mu$, say $\mu_{i-1}, \mu_i, \mu_{i+1}$ that satisfy conclusion 3) in Lemma~\ref{lem:unconstrained-scalar-reweighting} for $m = \pE_{\mu_i} x^2 \geq 1.$
\end{proof}

% \begin{lemma}
% Let $\mu$ be a pseudo-distribution over $\R$ satisfying $\{x \leq n\}$ and $\pE_{\mu}x^2 \geq 1.$  Then, for some absolute constant $C$, there exists a reweighting $\mu'$ of $\mu$ of degree $k \log^C{(n)}/\epsilon^2$ satisfying $\pE_{\mu} (x-m')^k \leq \epsilon^k m'^k$ for some $m'$ satisfying $|m'| \geq 1.$ 
% \end{lemma}
We are now ready to prove Lemma~\ref{lem:fixing-scalar-restated}.

\begin{proof}[Proof of Lemma~\ref{lem:fixing-scalar-restated}]
We know that $\mu$ satisfies $\{x^2 \leq n^2\}$ and $\pE_{\mu}x^2 \geq 1.$ We reweight by $x^{2\ell}$ where $\ell = O(d \log{(n)}/\epsilon)$ is obtained from Lemma \ref{lem:positive-scalar-reweighting}. Let $m \geq 1$ be such that the reweighted distribution $\mu'$ satisfies: $\pE_{\mu'} (x^2-m)^{2d} \leq \epsilon^{2d} m^{2d}.$ Or, $\pE_{\mu'} (x-\sqrt{m})^{2d} (x+\sqrt{m})^{2d} \leq \epsilon^{2d} m^{2d}.$

Now, since for every $x$, $(x-\sqrt{m})^{2d} + (x+\sqrt{m})^{2d} \geq \sqrt{m}^{2d}$, $\E (x-\sqrt{m})^{2d} + (x+\sqrt{m})^{2d} \geq \sqrt{m}^{2d}$ for every distribution over $\R$. Thus, for every distribution over $\R$, one of $\E (x-\sqrt{m})^{2d}$ and $\E (x+\sqrt{m})^{2d}$ is at least $0.5 \sqrt{m}^{2d}.$ Using Fact~\ref{fact:univariate},  the same conclusion holds for every pseudo-distribution of degree at least $2d+1$ and in particular, for $\mu'$. 

Thus, say $\pE_{\mu} (x-\sqrt{m})^{2d} \geq 0.5 \sqrt{m}^{2d}.$  Then, reweighting $\mu'$ by $(x-\sqrt{m})^{2d}$ to obtain $\mu''$ yields the $\pE_{\mu''} (x+\sqrt{m})^{2d} = \frac{\pE_{\mu'}  (x-\sqrt{m})^{2d} (x+\sqrt{m})^{2d}}{\pE_{\mu'}  (x-\sqrt{m})^{2d}} $ and $\pE_{\mu'} (x+\sqrt{m})^{2d}{\pE_{\mu'}(x-\sqrt{m})^{2d} } \leq \epsilon^{2d} (\sqrt{m})^{2d}$ proving the claim for $m' = -\sqrt{m}$ which is at least $1$ in magnitude.
\end{proof}

\section{Fixing vector-valued random variables}
\label{sec:fix-vector}

We show that distributions over the $d$-dimensional unit ball have $O(d)$-degree reweightings such that the resulting distribution is concentrated around a single vector.
Furthermore, the proof of this result also extends to pseudo-distribution of degree at least $O(d)$.

\begin{lemma}[Subspace Fixing Reweighting, Lemma \ref{lem:fixing-subspace}, restated] \label{lem:fixing-subspace-restated}
  % Let $\mu$ be a distribution over the unit ball of $\R^d$.
  % Let $C\ge 1$ be such that $\E_{\mu}\norm{x}^2 \ge d^{-C}$.
  % Then, $\mu$ has a degree-$k$ reweighting ("subspace fixing reweighting") $\mu'$, where $k=d \cdot (\log d)^{C'}$ for $C'=C'(C)$ only a function of $C$, such that
  % \[
  %   \Norm{\E_{\mu'(x)}x}^2 \ge 0.99 \E_{\mu(x)}\Norm{x}^2  \mper
  % \]
  % Further, the reweighting polynomial $p=\mu'/\mu$ can be found in time $2^{O(k)}$, has all coefficients upper bounded by $2^{O(k)}$ in the monomial basis, and satisfies $p(x) \leq k^{O(k)} \norm{x}^{k}$.
  % Moreover, the conclusions above hold even if $\mu$ is a pseudo-distribution of degree at least $k+2$.
  For every $C\geq 1$ there is some $C'$, such that if $\mu$ is a distribution over the unit ball $\{ x : \norm{x} \leq 1 \}$ of $\R^d$ such that  $\E_{\mu}\norm{x}^2 \ge d^{-C}$ then there is a degree $k= (d/\delta) \cdot (\log d)^{C'}$ reweighting $\mu'$ of $\mu$  such that
  \[
    \Norm{\E_{\mu'(x)}x}^2 \ge (1-\delta) \E_{\mu(x)}\Norm{x}^2  \mper
  \]
 Further, the reweighting polynomial $p=\mu'/\mu$ can be found in time $2^{O(k)}$, has all coefficients upper bounded by $2^{O(k)}$ in the monomial basis, and satisfies $p(x) \leq k^{O(k)} \norm{x}^{k}$.
Moreover, the result extends to pseudo-distributions $\mu$ of degree at least $d = k+2$, in which case, the reweighted pseudo-distribution $\mu'$ is of degree $d - k.$ 
\end{lemma}

On a high level, the proof goes as follows: 
The final reweighting is a combination of three reweightings.
The first reweighting approximately fixes the scalar variable $\norm{x}^2$ as in the previous section.
The second reweighting ensures that a single direction captures the expected norm in the sense that for some unit vector $v$ the variable $\iprod{v,x}^2$ has expectation close the expectation of $\norm{x}^2$ (which also means that the second moment is close to rank-1 in trace norm).
This step is the key innovation of this section. 
The final step is to fix the variable $\iprod{v,x}$ such that its expectation is approximately fixed to at least the square root of the expectation of $\iprod{v,x}^2$, which ensures that the norm of the expectation of $v$ is large. 

\begin{proof}
Let $\e>0$ be a small constant ($\delta/10$ would suffice, as the proof shows) to be set later and let $\mu$ be a pseudo-distribution over the unit ball of $\R^d$ that satisfies the requirements of the theorem. We will first find a reweighting $\mu'$ of $\mu$ such that $\pE_{\mu'(x)}\iprod{v,x}^2 \ge (1-\e) \pE_{\mu(x)}\norm{x}^2$.
  Then, by fixing the scalar variable $\iprod{v,x}$ as in the previous section, we obtain another reweighting that satisfies the requirements of the theorem.

  By fixing the scalar variable $\norm{x}^2$ as in the previous section, we may assume that $\mu$ satisfies $\pE_{\mu(x)} \norm{x}^{2k} \le (1+\e)^{k}\cdot (\pE_{\mu(x)}\norm{x}^2)^{k}$ for all $k\le (\e^{-1}\log d)^{O(1)}$.
  By \Holder's inequality for pseudo-expectations (this follows by an application of Fact~\ref{fact:univariate} to the scalar random variable $\norm{x}^2$), we also have $\pE_{\mu(x)} \norm{x}^{2k+2}\ge  \pE_{\mu(x)} \norm{x}^2\cdot\pE_{\mu(x)} \norm{x}^{2k}$.
  \Pnote{why do we lose $1-\epsilon$ here?}
%   \Dnote{TODO: check and add reference for this application of Holder}
% \Pnote{This Holder is easy right - since it's for a scalar random variable it follows from Fact~\ref{fact:univariate}?}  
 
 Let $v$ be a random unit vector in $\R^d$.
  It is known that for every $k\in\N$ and $x\in\R^n$,
  $\E_v \iprod{v,x}^{2k} = c_k \norm{x}^{2k}$ where $c_k$ is the $k$-th raw moment of the Beta distribution with parameter $1$ and $d-2$ \cite{MR644435-Stam82}.
  If $k\le O(d/\e)$ then these moments satisfy that $c_{k+1}\ge(1-\e) c_k$ and $1\ge c_k\ge k^{-O(k)}$.

  Using the bound $c_{k+1}\ge (1-\e)c_k$ and the fact that the variable $\norm{x}^2$ is approximately fixed, it follows that
  \begin{align}
    \E_v \pE _{\mu(x)} \iprod{v,x}^{2k} &= c_{k}\cdot\pE_{\mu(x)} \norm{x}^{2k}\,,
    \\
    \E_v \pE _{\mu(x)} \iprod{v,x}^{2k+2} & = c_{k+1}\cdot \pE_{\mu(x)} \norm{x}^{2k+2} \ge (1-\e)^2 c_k \pE_{\mu(x)} \norm{x}^{2} \cdot \pE_{\mu(x)}\norm{x}^{2k}\,. 
  \end{align}
  Combining the above two equations,
  \begin{align}
    \E_v \pE _{\mu(x)} \iprod{v,x}^{2k+2}  \ge (1-\e)^2 \cdot\pE_{\mu(x)} \norm{x}^{2} \cdot \E_v \pE _{\mu(x)} \iprod{v,x}^{2k} \,. 
  \end{align}  
  which means that there exists a unit vector $v$ such that
  \begin{math}
    \pE_{\mu(x)} \iprod{v,x}^{2k+2}\ge (1-\e)^2 \pE_{\mu(x)}\norm{x}^2 \cdot \E_v \pE _{\mu(x)} \iprod{v,x}^{2k}\,.
  \end{math}
  Therefore, the reweighting polynomial $x\mapsto \iprod{v,x}^{2k}/\pE_{\mu(x)}\norm{x}^{2k}$ yields a reweighting $\mu'$ such that $\pE_{\mu'(x)}\iprod{v,x}^2\ge (1-\e)^2\cdot \pE_{\mu(x)}\norm{x}^2$.

  It remains to argue that we can find such a unit vector in time $2^{O(k)}$.
  To that end we will show that a random unit vector succeeds with probability $2^{-O(k)}$.
  That argument will also show that the normalization factor $\pE_{\mu(x)}\norm{x}^{2k}$ is bounded by $k^{O(k)}$.

The key step is to bound the variance of the variable $\pE_{\mu(x)}\iprod{v,x}^{2k}$ (over the randomness of $v$).
  \begin{align}
    \E_{v}\Paren{\pE_{\mu(x)}\iprod{v,x}^{2k}}^2
    \le  \E_{v}\pE_{\mu(x)} \iprod{v,x}^{4k}
    = c_{2k} \cdot \pE_{\mu(x)}\norm{x}^{4k}
    \le k^{O(k)}\cdot \Paren{c_{k}\cdot \pE_{\mu(x)}\norm{x}^{2k}}
  \end{align}
  The first inequality is Cauchy--Schwarz for pseudo-expectations.
  The second inequality uses the bounds $c_{2k}\le k^{O(k)}c_k^2$ and $\pE_{\mu(x)}\norm{x}^{4k}\le (1+\e)^{2k} (\pE_{\mu(x)}\norm{x}^{2k})^2$ (the latter holds because we approximately fixed the scalar variable $\norm{x}^2$).

A standard Markov-like inequality (see for example \cite[Lemma 5.3]{DBLP:conf/stoc/BarakKS15}) shows that the following event over random unit vectors $v$ has probability at least $k^{-O(k)}$ (note that this probability is w.r.t. the distribution of the random variable $v$, which is an actual distribution),
  \begin{displaymath}
    \Set{
      \begin{gathered}
        \pE_{\mu(x)} \iprod{v,x}^{2k+2}
        \ge (1-\e)^3 \pE_{\mu(x)}\norm{x}^2 \cdot \pE_{\mu(x)} \pE_{\mu(x)} \iprod{v,x}^{2k}
        \\
        \pE_{\mu(x)} \iprod{v,x}^{2k}
        \ge k^{O(-k)} \cdot \pE_{\mu(x)} \norm{x}^{2k}
      \end{gathered}
    }
  \end{displaymath}
  Any unit vector that satisfies the above conditions yields a reweighting polynomial $p(x)\propto \iprod{v,x}^{2k}$ that satisfies the conclusion of the theorem for $\epsilon < \delta/10.$
\end{proof}

\section{Conclusions and further directions} \label{sec:conclusions}

We have shown an $\exp(\tilde{O}(\sqrt{n}/\epsilon^2))$ time algorithm for the $\BSS_{1,1-\epsilon}$ problem, or equivalently to the problem of finding (up to a small $\ell_2$ error) a rank one matrix in a subspace.
This raises several questions. 
One question, mentioned in Remark~\ref{rem:perfect-completeness} is whether we can remove the perfect, or near-perfect, completeness condition.  Another, possibly more important, question  is whether one can improve the exponent $\sqrt{n}$ further.
Indeed, it has been suggested (\cite{DBLP:journals/eccc/AaronsonIM14}) that the $\BSS_{c,s}$ problem for constants $s<c$ might have a quasi-polynomial time algorithm (which implies the conjecture that $QMA(2) \subseteq EXP$).
Given that our algorithm is inspired by Lovett's $\tilde{O}(\sqrt{n})$ bound on the communication complexity of rank $n$ matrices, it is tempting to speculate that the full log rank conjecture (i.e., a $\polylog(n)$ bound) would imply such a quasi-polynomial time algorithm.
We think a black box reduction from the log rank conjecture to such an algorithm is unlikely.
For starters, we would need the proof of the log rank conjecture to be embedded in the sos proof system.
But even beyond that, it seems that we need more general statements, that (unlike the log rank conjecture) do not  apply only to \textit{Boolean} matrices.
There do seem to be natural such statements that could imply improved algorithmic results.
In particular, we believe resolving the following two questions could help in making such progress:

\begin{question} \label{q:negative-reweight} Is it the case that for every distribution $\mu$ over $\sphere$ and every $\epsilon,\delta >0$ there is a (not necessarily positive) function $r:\sphere\rightarrow\R$ such that $\E_{v\sim\mu} |r(v)|=1$, $\E_{v\sim \mu} |r(v)| \log |r(v)|   \leq O(n^\delta)$ and a nonzero rank one $L$ such that
\[
\norm{\E_{v\sim \mu} [r(v) vv^\top] - L}_F \leq \epsilon\norm{L}_F \;\;\text{?}
\]
\end{question}

A positive solution for Question~\ref{q:negative-reweight} for any $\delta<1/2$ would be very interesting. 
It may\footnote{As mentioned in Footnote~\ref{fn:approx-monochromatic}, improving the bounds on the log rank conjecture might require better control of the  dependence of the bound on $\epsilon$ than we need for our setting.} improve the best known bound for the log rank conjecture to $\tilde{O}(n^\delta)$  and if  appropriately extended to pseudo-distributions, improve our algorithm's running time to $\exp(\tilde{O}(n^\delta))$ as well.
We do know that the answer to this question is \textit{No} if one does not allow \textit{negative} reweighting functions.

Another interesting question is the following:

\begin{question} \label{q:control2to4} 
Is there a function $f:\R_+\rightarrow\R_+$ such that for every $\delta>0$ and a distribution $\mu$ over $\sphere$,
there is an $O(n^\delta)$ round reweighting $\mu'$ of $\mu$ such that
\[
\E_{v,v' \sim \mu} \iprod{v,v'}^4 = f(\delta)\left(\E_{v,v' \sim \mu} \iprod{v,v'}^2 \right)^2 \;\;\text{?}
\]
\end{question}

We do not know of a way to use a positive  answer for Question~\ref{q:control2to4} for an improved bound on the log rank conjecture, but (an appropriate sos-friendly version of) it does imply an improved algorithm for the problem of ``$2$ vs $4$ provers QMA'' where, in the completeness case (i.e., when the state $\rho$ is accepted by the measurement $\cM$), there's a quantum proof given by a 4-partite separable state (i.e, four non-entangled provers can certify that $\rho$ is accepted by $\cM$) that the polynomial time quantum verifier accepts and in the soundness case (i.e, when $\tr(\cM \rho) < 1/3$), the verifier rejects any proof by four provers that can be split into two disjoint sets so that any shared entangled state is only between provers in the same set.
%This variant seems ``morally related'' to the \textit{small set expansion} problem (though we do not know of a reduction in either direction).

\section*{Acknowledgement}
We thank the anonymous reviewers for suggestions on improved presentation of the paper. We thank Vijay Bhattiprolu, Bill Fefferman, Cedric Lin, Anand Natarajan for pointing out typos and inaccuracies in a previous version of the paper and many illuminating comments. We thank Madhur Tulsiani for pointing out bugs in previous versions of the paper and several suggestions for improved presentation.
\addreferencesection
\bibliographystyle{amsalpha}
\newcommand{\etalchar}[1]{$^{#1}$}
\def\dbar{\leavevmode\hbox to 0pt{\hskip.2ex \accent"16\hss}d}
  \def\dbar{\leavevmode\hbox to 0pt{\hskip.2ex \accent"16\hss}d}
  \def\dbar{\leavevmode\hbox to 0pt{\hskip.2ex \accent"16\hss}d}
  \def\romsup#1{{\edef\next{\the\font}$^{\next#1}$}}
  \def\polhk#1{\setbox0=\hbox{#1}{\ooalign{\hidewidth
  \lower1.5ex\hbox{`}\hidewidth\crcr\unhbox0}}}
  \def\polhk#1{\setbox0=\hbox{#1}{\ooalign{\hidewidth
  \lower1.5ex\hbox{`}\hidewidth\crcr\unhbox0}}}
  \def\polhk#1{\setbox0=\hbox{#1}{\ooalign{\hidewidth
  \lower1.5ex\hbox{`}\hidewidth\crcr\unhbox0}}}
  \def\polhk#1{\setbox0=\hbox{#1}{\ooalign{\hidewidth
  \lower1.5ex\hbox{`}\hidewidth\crcr\unhbox0}}}
  \def\polhk#1{\setbox0=\hbox{#1}{\ooalign{\hidewidth
  \lower1.5ex\hbox{`}\hidewidth\crcr\unhbox0}}}
  \def\ocirc#1{\ifmmode\setbox0=\hbox{$#1$}\dimen0=\ht0 \advance\dimen0
  by1pt\rlap{\hbox to\wd0{\hss\raise\dimen0
  \hbox{\hskip.2em$\scriptscriptstyle\circ$}\hss}}#1\else {\accent"17 #1}\fi}
  \def\cfac#1{\ifmmode\setbox7\hbox{$\accent"5E#1$}\else
  \setbox7\hbox{\accent"5E#1}\penalty 10000\relax\fi\raise 1\ht7
  \hbox{\lower1.15ex\hbox to 1\wd7{\hss\accent"13\hss}}\penalty 10000
  \hskip-1\wd7\penalty 10000\box7}
  \def\cfac#1{\ifmmode\setbox7\hbox{$\accent"5E#1$}\else
  \setbox7\hbox{\accent"5E#1}\penalty 10000\relax\fi\raise 1\ht7
  \hbox{\lower1.15ex\hbox to 1\wd7{\hss\accent"13\hss}}\penalty 10000
  \hskip-1\wd7\penalty 10000\box7}
  \def\cfac#1{\ifmmode\setbox7\hbox{$\accent"5E#1$}\else
  \setbox7\hbox{\accent"5E#1}\penalty 10000\relax\fi\raise 1\ht7
  \hbox{\lower1.15ex\hbox to 1\wd7{\hss\accent"13\hss}}\penalty 10000
  \hskip-1\wd7\penalty 10000\box7}
  \def\ocirc#1{\ifmmode\setbox0=\hbox{$#1$}\dimen0=\ht0 \advance\dimen0
  by1pt\rlap{\hbox to\wd0{\hss\raise\dimen0
  \hbox{\hskip.2em$\scriptscriptstyle\circ$}\hss}}#1\else {\accent"17 #1}\fi}
  \def\ocirc#1{\ifmmode\setbox0=\hbox{$#1$}\dimen0=\ht0 \advance\dimen0
  by1pt\rlap{\hbox to\wd0{\hss\raise\dimen0
  \hbox{\hskip.2em$\scriptscriptstyle\circ$}\hss}}#1\else {\accent"17 #1}\fi}
  \def\polhk#1{\setbox0=\hbox{#1}{\ooalign{\hidewidth
  \lower1.5ex\hbox{`}\hidewidth\crcr\unhbox0}}}
  \def\ocirc#1{\ifmmode\setbox0=\hbox{$#1$}\dimen0=\ht0 \advance\dimen0
  by1pt\rlap{\hbox to\wd0{\hss\raise\dimen0
  \hbox{\hskip.2em$\scriptscriptstyle\circ$}\hss}}#1\else {\accent"17 #1}\fi}
  \def\polhk#1{\setbox0=\hbox{#1}{\ooalign{\hidewidth
  \lower1.5ex\hbox{`}\hidewidth\crcr\unhbox0}}}
  \def\cfudot#1{\ifmmode\setbox7\hbox{$\accent"5E#1$}\else
  \setbox7\hbox{\accent"5E#1}\penalty 10000\relax\fi\raise 1\ht7
  \hbox{\raise.1ex\hbox to 1\wd7{\hss.\hss}}\penalty 10000 \hskip-1\wd7\penalty
  10000\box7} \def\cfac#1{\ifmmode\setbox7\hbox{$\accent"5E#1$}\else
  \setbox7\hbox{\accent"5E#1}\penalty 10000\relax\fi\raise 1\ht7
  \hbox{\lower1.15ex\hbox to 1\wd7{\hss\accent"13\hss}}\penalty 10000
  \hskip-1\wd7\penalty 10000\box7} \def\cdprime{$''$}
  \def\polhk#1{\setbox0=\hbox{#1}{\ooalign{\hidewidth
  \lower1.5ex\hbox{`}\hidewidth\crcr\unhbox0}}}
  \def\cftil#1{\ifmmode\setbox7\hbox{$\accent"5E#1$}\else
  \setbox7\hbox{\accent"5E#1}\penalty 10000\relax\fi\raise 1\ht7
  \hbox{\lower1.15ex\hbox to 1\wd7{\hss\accent"7E\hss}}\penalty 10000
  \hskip-1\wd7\penalty 10000\box7}
  \def\ocirc#1{\ifmmode\setbox0=\hbox{$#1$}\dimen0=\ht0 \advance\dimen0
  by1pt\rlap{\hbox to\wd0{\hss\raise\dimen0
  \hbox{\hskip.2em$\scriptscriptstyle\circ$}\hss}}#1\else {\accent"17 #1}\fi}
\providecommand{\bysame}{\leavevmode\hbox to3em{\hrulefill}\thinspace}
\providecommand{\MR}{\relax\ifhmode\unskip\space\fi MR }
% \MRhref is called by the amsart/book/proc definition of \MR.
\providecommand{\MRhref}[2]{%
  \href{http://www.ams.org/mathscinet-getitem?mr=#1}{#2}
}
\providecommand{\href}[2]{#2}

\appendix

%!TEX root = ../subexpalg.tex
\section{Proof of Theorem \ref{thm:multi-dimensional-rank1}}
\label{sec:app-structure-thm}

\begin{proof}[Proof of Theorem \ref{thm:multi-dimensional-rank1}]
%We sketch the proof here as it is based on ideas that already appear in the proofs of Theorem~\ref{thm:structure-reals} and Corollary~\ref{cor:making-progress}.

Let $\mu$ be the pseudo-distribution on $(u,v) \in (\sphere)^2.$ As in the proof of Theorem~\ref{thm:structure-reals}, our final reweighting is obtained by combining a sequence of reweightings $\mu_0 = \mu, \mu_1, \mu_2, \ldots, \mu_T$. Let $m_1^j, m_2^j$ denote $\pE_{\mu_j} u, \pE_{\mu_j} v$, respectively and set $m_1^0 = m_1$ and $m_2^0 = m_2$. 

In the first step, for each $1 \leq i \leq 2$, we do the following: 
\begin{enumerate}
\item Let $S_1$ be the subspace spanned by the largest $\lceil 2 \sqrt{n} \rceil$ eigenvectors of $\E_{\mu(u)} \dyad{u}$. Define $S_2$ similarly for $v$. Reweight $\mu(u)$ (respectively, $\mu(v)$) in the subspace $S_1$ (respectively, $S_2$) using Lemma~\ref{lem:fixing-subspace} using a $\tilde{O}(\sqrt{n})$ degree SoS polynomial.
\item Reweight $\langle u, m_1 \rangle$ ($\langle v, m_2 \rangle$, respectively) using Lemma~\ref{lem:fixing-scalar} using $\tilde{O}(\sqrt{n}/\epsilon)$-degree.
\end{enumerate}

Similar to the proof of Corollary~\ref{cor:making-progress}, we can argue that at the end of the first step, $\norm{m_1}^2 \norm{m_2}^2 \geq \Theta(1/n).$ In any iterative step, we do the following. If $\norm{\dyad{m_1} - \E_{\mu(u)}\dyad{u}}_F > \epsilon \norm{\dyad{u}}$, 
\begin{enumerate}
\item Let $S_i$ be the subspace spanned by the largest $\lceil 2 \sqrt{n} \rceil$ eigenvectors of $\E_{\mu(u_i)} \dyad{u_i^{\perp}}$ where $u_i^{\perp}$ is the component of $u_i$ orthogonal to $m_i$. Reweight $\mu(u_i)$ in the subspace $S_i'$ - the span of $S_i$ and the direction $\E_{\mu(u_i)} \dyad{u_i^{\perp}}$ using Corollary~\ref{cor:making-progress} using $\tilde{O}(\sqrt{n}/\epsilon)$ degree.
\item Reweight $\langle u_i, m_i \rangle$ using Lemma~\ref{lem:fixing-scalar} using $\tilde{O}(\sqrt{n}/\epsilon)$-degree.
\end{enumerate}

We now track the potential function $\norm{m_1}^2 \norm{m_2}^2.$ In any step, the second reweighting above implies that under any reweighting $\norm{m_i}$ doesn't decrease by a factor of more than $1-\epsilon/10$. The first reweighting yields that at least one of $\norm{m_1}^2$ or $\norm{m_2}^2$ increases by a factor of $(1+\epsilon/2)$. In effect, after each reweighting, the potential rises by a multiplicative $(1+\Theta(\epsilon))$. Since $\norm{m_1}^2 \norm{m_2}^2 \leq 1$ and at least $\Theta(1/n)$ after the first step, the number of steps in the reweighting is upper bounded by $O(\log{(n)}/\epsilon)$ giving the result.
\end{proof}

% \section{Proof of Corollary~\ref{cor:structure}}\label{sec:app-complex}

\section{Reduction Between Real and Complex Best Separable State Problems}
\label{app:red-complex-to-real}
\begin{lemma}
For every subspace $\cW \subseteq \C^{n^2}$, there's a subspace $\cY \subseteq \R^{4n^2}$ such that:
\begin{enumerate}
\item Completeness: If there's an $x,y \in \sphere(\C)$ such that $xy^{*} \in \cW$, then there's a $u,v \in \bbS^{2n-1}$ such that $uv^{*} \in \cY.$
\item Soundness: If there's a $U \in \cY$ and $u_0,v_0$ such that $\norm{ u_0\transpose{v_0} - U }_F\leq \epsilon \norm{u_0\transpose{v_0}}_F$, then there's a $X \in \cW$ and an $x_0,y_0$ such that $\norm{ x_0y_0^{*}- X }_F\leq \epsilon \norm{x_0y_0^{*}}_F.$
\end{enumerate}
\end{lemma}

\begin{proof}
It is easiest to describe the construction of the subspace $U$ from $\cW$ in two steps. Let $\langle W^j, X \rangle = 0$ for $j \leq \codim(\cW)$ be the linear constraints that define $\cW$. Write $X = A + i B $ for $i= \sqrt{-1}$ and $W^j = C^j + i D^j$. Then, $X \in \cW$ iff for every $j \leq \codim(\cW)$,

  \begin{displaymath}
    \Set{
      \begin{gathered}
 \langle C^j, A \rangle + \langle D^j, B \rangle = 0 \\
       \langle D^j, A\rangle - \langle C^j, B \rangle  = 0.
      \end{gathered}
    }\label{eq:constraints-W-1}
  \end{displaymath}
\end{proof}

Let $\cW' \subseteq \R^{n \times n} \times \R^{n \times n} $ be the subspace of ordered pairs of $n \times n$ matrices $(A,B)$ satisfying \eqref{eq:constraints-W-1} for each $j \leq \dim(\cW).$

Observe that by construction, a matrix $X = A + iB \in \cW$ iff the pair $(A,B) \in \cW'$. Next, we define a subspace $\cY \subseteq \R^{2n \times 2n}$ as follows. We think of each $Y \in \cY$ as a $2 \times 2$ block matrix of $n\times n$ matrices with the blocks being labeled as $Y_{11}, Y_{12}, Y_{21}, Y_{22}$ in the natural way. We define $\cY$ 
\begin{equation}
Y \in \cY \text{ iff } (Y_{11} + Y_{22}, Y_{21} - Y_{12}) \in \cW'. \label{eq:def-cY}
\end{equation}

We now claim that the subspace $\cY$ satisfies the requirements of the Lemma. First observe that if $Y \in \cY$, then by our construction, $(Y_{11} + Y_{22}, Y_{21} - Y_{12}) \in \cW'$ and consequently, 
\begin{equation}
X = (Y_{11} + Y_{22}) + i (Y_{21} - Y_{12}) \in \cW. \label{eq:inverse-transform}
\end{equation}

\paragraph{Completeness} If $xy^{*} \in \cW$ then, writing $x = u + i v$ and $y = u' + i v'$ and setting $A = u\transpose{u'} + v \transpose{v'} $ and $B = v \transpose{u'} - u \transpose{v'}$ yields that $(A,B) \in \cW'$ and thus, consequently, $Y = (u,v) \transpose{(u',v')} \in \cY.$ 

\paragraph{Soundness} Suppose $Y \in \cY$ and there's $u,v$ such that $\norm{u\transpose{v} - Y}_F \leq \epsilon \norm{u \transpose{v}}.$ Let $u_1, u_2$ ($v_1,v_2$) be the components of $u$ in the first and second column (row) blocks respectively. From \eqref{eq:inverse-transform}, we know that $X = A +iB$ for $A = (Y_{11} + Y_{22}) + i (Y_{21} - Y_{12}) \in \cW'$. Let $U = u_1 + i u_2$ and $V = v_1 + i v_2$. Then, we can rewrite the above as:

\[
X = A + iB = (u_1 + i u_2) (v_1 + i v_2)^{*} + (Y_{11} + Y_{22} - u_1\transpose{v_1} - u_2 \transpose{v_2}) + i (Y_{21} - Y_{22} - u_2 \transpose{v_2} + u_1 \transpose{v_1}) \in \cW.
\]

Now, $\norm{(u_1 + i u_2) (v_1 + i v_2)^{*}}^2 = \norm{u_1}^2 + \norm{u_2}^2 + \norm{v_1}^2 + \norm{v_2}^2.$

And by an application of triangle inequality, 
\begin{align*}
\norm{(Y_{11} + Y_{22} - u_1\transpose{v_1} - u_2 \transpose{v_2}) + i (Y_{21} - Y_{22} - u_2 \transpose{v_2} + u_1 \transpose{v_1})} &\leq \sum_{s,t = 1}^2 \norm{Y_{st} - u_s \transpose{v_t}}\\
&\leq \norm{u\transpose{v} - Y}_F \leq \epsilon \norm{u \transpose{v}}\\
&\leq \epsilon \norm{U} \norm{V}\\
&= \epsilon\norm{UV^{*}}. 
\end{align*}

\section{Higher Rank Structure Theorem}
\label{sec:higher-rank-structure-theorem}
\begin{theorem}
\label{thm:multi-dimensional-rank1-app}
Let $\e>0$, let $\mu$ be a pseudo-distribution over $(u_1,u_2,\ldots, u_r)$ such that $\sum_{i} \norm{u_i}^2 = 1$. Let the degree of $\mu$ be at least $k+2$, where $k= \sqrt{rn} (\log n)^C/\epsilon^2$ for an absolute constant $C\ge 1$.
Then, $\mu$ has a degree-$k$ reweighting $\mu'$ such that for each $1 \leq j \leq r$ 
\[
\sum_{i = 1}^r\Norm{\dyad{m_i} -\pE_{\mu'(u_i)} \dyad u_i}^2_F \leq \epsilon^2 \cdot \left(\sum_{i = 1}^r \Norm{\dyad{m_i}}^2_F\right) \mcom
\]
where $m_i = \E_{\mu'(u_i)}u_i \mper$
Furthermore, we can find the reweighting polynomial $p=\mu'/\mu$ in time $2^{O(k)}$ and $p$ has only rational coefficients in the monomial basis with numerators and denominators of magnitude at most $2^{O(k)}$.
\end{theorem}
\begin{remark}
The above theorem can be generalized  to design an algorithm that finds symmetric rank $r$ matrices of Frobenius norm $1$ inside subspaces $\cW$ in time $2^{\tilde{O}(\sqrt{rn}/\epsilon^2)}.$
\end{remark}
\begin{proof}
For every $(u_1,u_2, \ldots, u_r)$ in the support of $\mu$, define $u \in \R^{rn}$ be the concatenation of $u_1, u_2, \ldots, u_r$. Then, $\mu$ can be equivalently thought of as a distribution over the unit sphere. Then, $\norm{u}^2 = 1.$ By Theorem~\ref{thm:structure-reals}, there's a $\tilde{O}(\sqrt{rn}/\epsilon^2)$ degree reweighting $\mu'$ of $\mu$ such that there exists a $u^0$ such that $\norm{\dyad{u^0} - \pE_{\mu'}\dyad{u} }_F \leq \epsilon \norm{\dyad{u^0}}_F.$

Let $m_i = \pE_{\mu'} \dyad{(u^0)_i}.$ Then, $\sum_{i =1}^r \norm{m_i}^2 = \norm{u^0}^2.$ On the other hand, $\norm{\pE_{\mu} \dyad{u_i} - \dyad{m_i}}_F^2 \leq \norm{\dyad{u} - \dyad{u^0}}_F^2 \leq \epsilon^2 \norm{\dyad{u^0}}^2_F =\epsilon^2  \sum_{i =1}^r \norm{m_i}^2.$

This completes the proof.
\end{proof}

\end{document}